\documentclass[aip,jmp,reprint,onecolumn]{revtex4-1}
\usepackage{amssymb,amsmath}

\newtheorem{proposition}{Proposition}
\newtheorem{lemma}{Lemma}

\newtheorem{corollary}{Corollary}
\newenvironment{proof}
{\par\noindent{\bf Proof.}} {\hfill$\scriptstyle\blacksquare$}

\begin{document}

\title {\bf Quantum informational properties of the Landau--Streater channel}

\author{Sergey N. Filippov}

\affiliation{Moscow Institute of Physics and Technology,
Institutskii Per. 9, Dolgoprudny, Moscow Region 141700, Russia}

\affiliation{Institute of Physics and Technology of the Russian
Academy of Sciences, Nakhimovskii Pr. 34, Moscow 117218, Russia}

\affiliation{Steklov Mathematical Institute of Russian Academy of
Sciences, Gubkina St. 8, Moscow 119991, Russia}

\author{Ksenia V. Kuzhamuratova}

\affiliation{Moscow Institute of Physics and Technology,
Institutskii Per. 9, Dolgoprudny, Moscow Region 141700, Russia}

\begin{abstract}
We study the Landau--Streater quantum channel $\Phi:
\mathcal{B}(\mathcal{H}_d) \mapsto \mathcal{B}(\mathcal{H}_d)$,
whose Kraus operators are proportional to the irreducible unitary
representation of $SU(2)$ generators of dimension $d$. We
establish $SU(2)$ covariance for all $d$ and $U(3)$ covariance for
$d=3$. Using the theory of angular momentum, we explicitly find
the spectrum and the minimal output entropy of $\Phi$. Negative
eigenvalues in the spectrum of $\Phi$ indicate that the channel
cannot be obtained as a result of Hermitian Markovian quantum
dynamics. Degradability and antidegradability of the
Landau--Streater channel is fully analyzed. We calculate classical
and entanglement-assisted capacities of $\Phi$. Quantum capacity
of $\Phi$ vanishes if $d=2,3$ and is strictly positive if $d
\geqslant 4$. We show that the channel $\Phi \otimes \Phi$ does
not annihilate entanglement and preserves entanglement of some
states with Schmidt rank $2$ if $d \geqslant 3$.
\end{abstract}

\maketitle

\section{Introduction}

The state of a finite dimensional quantum system is described by
the density operator $\varrho \in \mathcal{B}(\mathcal{H}_d)$
acting on the Hilbert space $\mathcal{H}_d$, $d = {\rm
dim}\mathcal{H}_d$. The density operator is Hermitian, positive
semidefinite, and has unit trace.

In the theory of open quantum systems, the most general form of
the density operator transformation due to its own evolution and
interaction with some environment (initially uncorrelated from the
system) is given by a quantum channel
$\Phi:\mathcal{B}(\mathcal{H}_{d_1}) \mapsto
\mathcal{B}(\mathcal{H}_{d_2})$, which is a completely positive
and trace preserving linear map (Ref.~\cite{holevo-2012}, section
6.3). In what follows, we consider the case when the system
dimension does not change, i.e., $d_1 = d_2 = d$.

There is a special class of unital quantum channels, which
preserve the maximally mixed state $\frac{1}{d} I$, where $I$ is
the identity operator, i.e. $\Phi[\frac{1}{d} I] = \frac{1}{d} I$.
The seminal result of Landau and Streater~\cite{Landau-Streater}
is that all unital channels $\Phi:\mathcal{B}(\mathcal{H}_2)
\mapsto \mathcal{B}(\mathcal{H}_2)$ are random unitary, i.e.
$\Phi[\varrho] = \sum_i p_i U_i \varrho U_i^{\dag}$ for some
probability distribution $\{p_i\}$ and unitary operators $U_i$
acting on $\mathcal{H}_2$; whereas for larger dimensions this is
not the case~\cite{Landau-Streater}. Moreover, Landau and
Streater~\cite{Landau-Streater} provided an example of quantum
channel $\Phi:\mathcal{B}(\mathcal{H}_d) \mapsto
\mathcal{B}(\mathcal{H}_d)$, which is unital but not random
unitary for all $d \geqslant 3$.

The main goal of this paper is to explore quantum informational
properties of the Landau--Streater map
\begin{equation} \label{LS-map}
\Phi[\rho] = \frac{1}{j(j+1)} \left( J_x \rho J_x + J_y \rho J_y +
J_z \rho J_z \right)
\end{equation}

\noindent defined through the $SU(2)$ generators $J_x$, $J_y$,
$J_z$ acting on a $(2j+1)$-dimensional Hilbert space
$\mathcal{H}_{2j+1}$. Physically, this space corresponds to the
state space of a spin-$j$ particle. Hermitian operators $J_x$,
$J_y$, $J_z$ are spin projection operators onto axes $x$, $y$,
$z$, respectively. Hereafter, we will use indices $x,y,z$ and
$1,2,3$ interchangeably. The map \eqref{LS-map} is completely
positive as it has diagonal sum (Kraus) representation
(Ref.~\cite{holevo-2012}, Corollary 6.13), and is trace preserving
as $\sum_{k=1}^3 J_k^2 = j(j+1) I$ (Ref.~\cite{Varshalovich},
section 6.1.2, formula (5)). The latter formula is also
responsible for unitality of the map~\eqref{LS-map}.

If $j=\frac{1}{2}$, then $J_k = \frac{1}{2} \sigma_k$, where
$\{\sigma_k\}_{k=1}^3$ is the conventional set of Pauli operators.
In this case $d=2$ and $\Phi$ is random
unitary~\cite{Landau-Streater}. Such a channel $\Phi$ is also
referred to as the best physical approximation of the universal
NOT gate for qubits~\cite{buzek-1999}.

If $j \geqslant 1$, then the map \eqref{LS-map} is an extremal
channel in the set of all channels, and therefore it cannot be
random unitary~\cite{Landau-Streater}. Also, in contrast to the
case $j = \frac{1}{2}$, the channel~\eqref{LS-map} differs from
the spin polarization-scaling channels~\cite{fm-2018} if $j
\geqslant 1$.

Since the Landau--Streater channel \eqref{LS-map} is essentially
unexplored, the goal of this paper is to fill the gap in both the
fundamental properties (such as covariance and spectrum of the
output states) and more specific quantum informational properties
(such as classical and quantum capacities, entanglement dynamics).

The paper is organized as follows.

In section~\ref{section-covariance}, we study covariance
properties of the Landau--Streater channel. In
section~\ref{section-spectrum-map}, we explicitly find the
spectrum of $\Phi$ for all $j$. In
section~\ref{section-spectrum-output}, we analyze the spectrum of
the output operator $\Phi[X]$ and reveal its peculiarities in the
case $j=1$. Such peculiarities are attributed to the fact that the
Landau--Streater channel for $j=1$ reduces to the Werner--Holevo
channel~\cite{wh-2002}. In section~\ref{section-multiplicativity},
we explicitly find the maximal $p$-norm and the minimal output
entropy of the general Landau--Streater channel. In
section~\ref{section-complementary}, physical realization of the
Landau--Streater channel and its complementary version is
discussed. In section~\ref{section-capacities}, different
capacities of the Landau--Streater channel are evaluated. In
section~\ref{section-entanglement}, we examine the entanglement
annihilation property of the channel $\Phi \otimes \Phi$. In
section~\ref{section-conclusions}, brief conclusions and outlook
are presented.

\section{Covariance} \label{section-covariance}

Following Refs.~\cite{holevo-1993}, consider a group $G$ and a
unitary representation $g \rightarrow U_g$, $g \in G$, in
$\mathcal{H}_d$. The channel $\Phi:\mathcal{B}(\mathcal{H}_d)
\mapsto \mathcal{B}(\mathcal{H}_d)$ is called covariant with
respect to representation $U_g$ if there exists a unitary
representation $g \rightarrow V_g$, $g \in G$, in $\mathcal{H}_d$
such that
\begin{equation}
\Phi[U_g X U_g^{\dag}] = V_g \Phi[X] V_g^{\dag}
\end{equation}

\noindent for all $g \in G$ and $X \in
\mathcal{B}(\mathcal{H}_d)$. Covariance of a channel has many
implications, e.g. the strong converse property of the
entanglement-assisted classical capacity~\cite{datta-2016}. Many
other properties and the structure of irreducibly covariant
quantum channels are reviewed in Ref.~\cite{mozrzymas-2017}.

We start with the simple observation that the Landau--Streater
channel \eqref{LS-map}, by construction, is endowed with the
$SU(2)$ covariance.

\begin{proposition} \label{prop-SU2-cov}
The Landau--Streater channel $\Phi:
\mathcal{B}(\mathcal{H}_{2j+1}) \rightarrow
\mathcal{B}(\mathcal{H}_{2j+1})$ is covariant with respect to the
unitary representation of $SU(2)$ for all spins $j$, namely,
$\Phi[U_g X U_g^{\dag}] = U_g \Phi[X] U_g^\dag$ for all $g \in
SU(2)$ and $X \in \mathcal{B}(\mathcal{H}_{2j+1})$.
\end{proposition}
\begin{proof}
Up to an irrelevant phase, any $U_g$, $g \in SU(2)$ can be
expressed through the $SU(2)$ generators
$\{J_{\alpha}\}_{\alpha\in\{x,y,z\}}$ as follows:
\begin{equation}
U_g = \exp \left( - i \theta \sum_{\alpha} n_{\alpha} J_{\alpha}
\right),
\end{equation}

\noindent where ${\bf n} = (n_x, n_y, n_z)$ is a unit vector in
$\mathbb{R}^3$, which defines the rotation axis, and $\theta \in
\mathbb{R}$ is the rotation angle. The operators $J_x,J_y,J_z$
satisfy the commutation relation $[J_{\alpha},J_{\beta}] = i
\sum_{\gamma=1}^3 e_{\alpha\beta\gamma} J_{\gamma}$, where
$e_{\alpha\beta\gamma}$ is the Levi-Civita symbol
(Ref.~\cite{Varshalovich}, section 2.1.2, formula (7)). Using such
a commutation relation, it is not hard to see that
(Ref.~\cite{Varshalovich}, section 3.1.3, formula (11); section
3.1.6, item (a); section 1.4.5, formula (33))
\begin{equation} \label{rotation}
U_g^{\dag} J_{\alpha} U_g = \sum_{\beta} Q_{\alpha\beta}
J_{\beta},
\end{equation}

\noindent where $(Q_{\alpha\beta})$ is the orthogonal matrix
describing the rotation in $\mathbb{R}^3$ around the axis ${\bf
n}$ by angle $\theta$ (Ref.~\cite{Varshalovich}, section 1.4.2).
Finally, we get
\begin{eqnarray}
\Phi[U_g X U_g^{\dag}] &=& \frac{1}{j(j+1)} \sum_{\alpha}
J_{\alpha} U_g X U_g^{\dag} J_{\alpha} \nonumber\\
&=& \frac{1}{j(j+1)} \sum_{\alpha} U_g (U_g^{\dag} J_{\alpha} U_g)
X (U_g^{\dag} J_{\alpha} U_g) U_g^{\dag} \nonumber\\
&=& \frac{1}{j(j+1)} \sum_{\alpha,\beta,\gamma} Q_{\alpha\beta}
Q_{\alpha\gamma} U_g J_{\beta}
X J_{\gamma} U_g^{\dag} \nonumber\\
&=& \frac{1}{j(j+1)} \sum_{\beta\gamma} \delta_{\beta\gamma} U_g
J_{\beta} X J_{\gamma} U_g^{\dag} \nonumber\\
&=& U_g \Phi[X] U_g^{\dag},
\end{eqnarray}

\noindent where we have taken into account that $\sum_{\alpha}
Q_{\alpha\beta} Q_{\alpha\gamma} = \delta_{\beta\gamma}$, the
Kronecker delta.
\end{proof}

Since the Landau--Streater channel is extreme in the set of all
channels~\cite{Landau-Streater} it follows that it is also an
extreme point of $SU(2)$ irreducibly covariant channels
(abbreviated as an EPOSIC channel). The general properties of an
EPOSIC channel such as the Kraus representation, the Choi matrix,
and the dual channel are reviewed in Ref.~\cite{nuwairan-2014}.

It turns out that in the case $j=1$ the Landau--Streater channel
is not only $SU(2)$ covariant, but also globally unitarily
covariant. It means that in dimension $d=3$ the
channel~\eqref{LS-map} possesses $U(3)$ covariance.

\begin{proposition} \label{prop-U3-cov}
In the case $j=1$, the Landau--Streater channel is globally
unitarily covariant, namely, for all $X \in
\mathcal{B}(\mathcal{H}_3)$ and unitary operators $U$ on
$\mathcal{H}_3$ the following holds
\begin{equation} \label{U3-covariance}
\Phi[U X U^\dag] = V \Phi[X] V^\dag,
\end{equation}

\noindent where the unitary operator $V$ is expressed through $U$
in the orthonormal basis of eigenvectors of $J_z$ via formula
\begin{equation} \label{UV-relation-U3}
U =
\begin{pmatrix}
u_{11} & u_{12} & u_{13} \\
u_{21} & u_{22} & u_{23} \\
u_{31} & u_{32} & u_{33}
\end{pmatrix},
\qquad
V =
\begin{pmatrix}
\overline{u_{33}} & -\overline{u_{32}} & \overline{u_{31}}
\\
- \overline{u_{23}} & \overline{u_{22}} & -\overline{u_{21}}
\\
\overline{u_{13}} & -\overline{u_{12}} & \overline{u_{11}}
\end{pmatrix}.
\end{equation}
\end{proposition}

\begin{proof}
Note that $V$ is unitary if $U$ is unitary. In the basis of
eigenvectors of $J_z$, the operators
$\{J_{\alpha}\}_{\alpha=x,y,z}$ have the following form if $j=1$:
\begin{equation}
J_x=
\begin{pmatrix} 0&\frac{1}{\sqrt{2}}&0\\
\frac{1}{\sqrt{2}}&0&\frac{1}{\sqrt{2}}\\0&\frac{1}{\sqrt{2}}&0\end{pmatrix},
\quad J_y=
\begin{pmatrix}0&-\frac{ i }{\sqrt{2}}&0\\ \frac{i}{\sqrt{2}}&0&-\frac{i}{\sqrt{2}}\\0&\frac{i}{\sqrt{2}}&0\end{pmatrix},
\quad J_z=
\begin{pmatrix}1&0&0\\0&0&0\\0&0&-1\end{pmatrix}.
\end{equation}

\noindent Substituting these operators in equation~\eqref{LS-map},
the direct calculation justifies the validity of formulas
\eqref{U3-covariance}, \eqref{UV-relation-U3}.
\end{proof}

The global unitary covariance is known to be a peculiar property
of the tracing, transposition, and identity maps. This feature
allowed one to find specific results for the Werner--Holevo
channel~\cite{datta-2004} and transpose-depolarizing
channels~\cite{datta-2006} as well as to prove additivity of
classical capacity for depolarizing quantum
channels~\cite{king-depol-2003}. We have just found out that the
Landau--Streater map for $j=1$ is globally unitarily covariant
too. However, as we show below, in the case $j>1$ the
Landau--Streater channel loses the property of global unitary
covariance.

\begin{proposition}
The Landau--Streater channel is not globally unitarily covariant
if $j>1$.
\end{proposition}
\begin{proof}
We prove the statement by constructing a counterexample. Suppose
the Landau--Streater channel $\Phi:
\mathcal{B}(\mathcal{H}_{2j+1}) \mapsto
\mathcal{B}(\mathcal{H}_{2j+1})$, $j>1$, is covariant with respect
to representation of $U(2j+1)$. Then for any unitary operator $U$
acting on $\mathcal{H}_{2j+1}$, there exists a unitary operator
$V$ such that
\begin{equation} \label{UV-assumption}
\Phi[U \rho U^\dag] = V \Phi[\rho] V^\dag
\end{equation}

\noindent holds true for all density operators $\rho$. This
implies that the output density operators $\Phi[U \rho U^\dag]$
and $\Phi[\rho]$ have identical spectra.

Consider eigenvectors of the spin projection onto $z$-axis
(Ref.~\cite{Varshalovich}, section 6.1.2, formula(5)):
\begin{equation}
J_z |j,m\rangle = m |j,m\rangle, \quad m = j, j-1, \ldots, -j.
\end{equation}

Let $\rho = |j,j\rangle\langle j,j|$ and
\begin{equation}
U = |j,j-1\rangle\langle j,j| + |j,j\rangle\langle j,j-1| +
\sum_{k=-j}^{j-2} |j,k\rangle\langle j,k|,
\end{equation}

\noindent then $\Phi[|j,j-1\rangle\langle j,j-1|]$ and
$\Phi[|j,j\rangle\langle j,j|]$ must have the same spectra.

On the other hand, action of the Landau--Streater map on the
states $|j,m\rangle \langle j,m|$ with definite spin projection
$m$ onto $z$-axis can be expressed explicitly by introducing
auxiliary operators $J_{\pm} = J_x \pm i J_y$ satisfying
(Ref.~\cite{Varshalovich}, section 2.3.3, formula (7); section
6.1.2, formula (13))
\begin{equation} \label{J-pm}
J_{\pm} | j,m \rangle = \sqrt{( j \mp m)(j \pm m + 1)} |j, m \pm 1
\rangle.
\end{equation}

\noindent Since $\Phi[X] = [j(j+1)]^{-1} \left( \frac{1}{2}J_- X
J_+ + \frac{1}{2} J_+ X J_- + J_z X J_z \right)$, we get
\begin{eqnarray} \label{Phi-on-jm}
\Phi[|j,m\rangle\langle j,m|] &=& \frac{1}{j(j+1)} \bigg[
\frac{1}{2}\left( j(j+1)-m(m-1) \right) |j,m-1\rangle\langle
j,m-1| \nonumber\\
&& + \frac{1}{2}\left( j(j+1)-m(m+1) \right) |j,m+1\rangle\langle
j,m+1| + m^2|j,m\rangle\langle j,m| \bigg],
\end{eqnarray}

\noindent from which we make conclusion about the spectrum of the
output state $\Phi[|j,m\rangle\langle j,m|]$:
\begin{equation} \label{spec-phi-jm}
{\rm Spec}\left( \Phi[|j,m\rangle\langle j,m|] \right) = \left\{
\frac{j(j+1)-m(m+1)}{2j(j+1)}, \frac{j(j+1)-m(m-1)}{2j(j+1)},
\frac{m^2}{j(j+1)}, 0,0,\ldots \right\}.
\end{equation}

\noindent If $j>1$, then ${\rm Spec}\left( \Phi[|j,j\rangle\langle
j,j|] \right) \neq {\rm Spec}\left( \Phi[|j,j-1\rangle\langle
j,j-1|] \right)$. This contradiction concludes the proof.
\end{proof}

\section{Spectral properties}

\subsection{Spectrum of the map} \label{section-spectrum-map}

The Landau--Streater channel $\Phi$ is Hermitian as it coincides
with its dual $\Phi^{\dag}$, therefore its spectrum
$\{\lambda_k\}_{k=0}^{(2j+1)^2-1}$ is real. Hermitian
eigenoperators $X_k$ satisfy $\Phi[X_k] = \lambda_k X_k$. Due to
unitality of $\Phi$, the identity operator $I$ is the
eigenoperator, so we can fix the corresponding eigenvalue
$\lambda_0 = 1$ for all $j$. By determinant ${\rm det}\Phi$ of the
channel $\Phi$ we will understand the product of its eigenvalues
$\prod_k \lambda_k$.

If $j=\frac{1}{2}$, then $J_x$, $J_y$, $J_z$ are eigenoperators of
$\Phi$ and $\lambda_1 = \lambda_2 = \lambda_3 = - \frac{1}{3}$. In
this case, ${\rm det}\Phi = - \frac{1}{27} < 0$, so the channel
$\Phi$ is not infinitesimally divisible~\cite{wolf-2008} and
cannot be obtained as a result of Markovian evolution, although it
can be realized physically, e.g., via collision
models~\cite{fpmz-2017}.

If $j=1$, then $J_x$, $J_y$, $J_z$ are eigenoperators of $\Phi$
with corresponding eigenvalues $\lambda_1 = \lambda_2 = \lambda_3
= \frac{1}{2}$. Five more eigenoperators have the form $3 \left(
\sum_{\alpha} n_{\alpha}^{(k)}J_{\alpha} \right)^2 - 2I$, ${\bf
n}^{(k)} \in \mathbb{R}^3$, $k=1,\ldots,5$ (Ref.~\cite{fm-2010},
formula (8) and text after formula (36)) and correspond to
eigenvalues $\lambda_4 = \ldots = \lambda_8 = - \frac{1}{2}$.
Similarly, ${\rm det}\Phi < 0$, so such a channel cannot be a
result of Markovian evolution.

In what follows, we find spectrum of the Landau--Streater map
$\Phi: \mathcal{B}(\mathcal{H}_{2j+1}) \mapsto
\mathcal{B}(\mathcal{H}_{2j+1})$ for an arbitrary integer or
half-integer $j$. As we show, the eigenoperators of $\Phi$ are
particularly related with the irreducible tensor operator
$T_{LM}^{(j)}$ for the $SU(2)$ group, which is also known as the
polarization operator (Ref.~\cite{Varshalovich}, section 2.4.2,
formula (6); section 8.4.3, formula (10)):
\begin{equation}
T_{LM}^{(j)} = \sqrt{\frac{2L+1}{2j+1}} \, \sum_{m_1,m_2=-j}^j
C_{j m_1 L M}^{j m_2} | j m_2 \rangle \langle j m_1 | =
\sum_{m_1,m_2 = -j}^{j} (-1)^{j-m_1} C_{j m_2 j -m_1}^{L M} | j
m_2 \rangle \langle j m_1 |,
\end{equation}

\noindent where $C_{j_1 m_1 j_2 m_2}^{J M}$ is the conventional
Clebsch--Gordan coefficient.

\begin{proposition} \label{prop-spectrum-LS}
The spectrum of the Landau--Streater map $\Phi:
\mathcal{B}(\mathcal{H}_{2j+1}) \mapsto
\mathcal{B}(\mathcal{H}_{2j+1})$ comprises $(2L+1)$-fold
degenerate eigenvalues
\begin{equation} \label{lambda-L}
\lambda_L = 1 - \frac{L(L+1)}{2j(j+1)}, \qquad L = 0, 1, \ldots,
2j.
\end{equation}

\noindent The corresponding eigenoperators are linearly
independent operators of the form $U_g T_{L0}^{(j)} U_g^{\dag}$,
where the operators $U_g$ belong to the unitary representation of
the $SU(2)$ group.
\end{proposition}
\begin{proof}
We start with the observation that $T_{L0}^{(j)}$ is the
eigenoperator of $\Phi$. To prove this fact we rewrite the
Landau--Streater channel in the form $\Phi[X] = [j(j+1)]^{-1}
(\frac{1}{2} J_+ X J_- + \frac{1}{2} J_- X J_+ + J_z X J_z)$ and
use the commutation relations $[J_{\pm},T_{L0}^{(j)}] =
\sqrt{L(L+1)} T_{L \pm 1}^{(j)}$ and $[J_z,T_{L0}^{(j)}] = 0$ (see
Ref.~\cite{Varshalovich}, section 2.4.1, formula (1); section
2.3.3, formula (7); for the Clebsch-Gordan coefficients $C_{L0 1
\pm 1}^{L \pm 1} = \mp \frac{1}{\sqrt{2}}$ and $C_{L0 10}^{L0} =
0$ see section 8.5.1, formula (8)). We get
\begin{equation} \label{Phi-on-TL0}
\Phi \left[ T_{L0}^{(j)} \right] = \frac{1}{j(j+1)} \left[ \left(
\frac{1}{2} J_+ J_- + \frac{1}{2} J_- J_+ + J_z J_z \right)
T_{L0}^{(j)} - \frac{\sqrt{L(L+1)}}{2} \left( J_+ T_{L-1}^{(j)} +
J_- T_{L1}^{(j)} \right) \right].
\end{equation}

\noindent The first expression in parentheses $(\cdot)$ is
$j(j+1)I$, whereas the second expression in parentheses $(\cdot)$
can be simplified because $J_{\pm} = \mp
\sqrt{\frac{2j(j+1)(2j+1)}{3}} T_{1 \pm 1}^{(j)}$
(Ref.~\cite{Varshalovich}, section 2.4.2, formula (10); section
2.3.3, formula (7)) and the product $T_{1 \pm 1}^{(j)} T_{L
 \mp 1}^{(j)}$ is known (Ref.~\cite{Varshalovich}, section 2.4.4,
formula (16)):
\begin{eqnarray}
&& J_+ T_{L-1}^{(j)} + J_- T_{L1}^{(j)} = -
\sqrt{\frac{2j(j+1)(2j+1)}{3}} \left( T_{1 1}^{(j)} T_{L
 -1}^{(j)} - T_{1 -1}^{(j)} T_{L
 1}^{(j)} \right) \nonumber\\
&& = - \sqrt{\frac{2j(j+1)(2j+1)}{3}} \sum_{L'} (-1)^{2j+L'}
\sqrt{3(2L+1)} \left\{
\begin{array}{ccc}
  1 & L & L' \\
  j & j & j \\
\end{array} \right\} \left( C_{11L-1}^{L'0} - C_{1-1L1}^{L'0}
\right) T_{L'0}^{(j)}. \nonumber\\
&& \label{difference-T}
\end{eqnarray}

\noindent The Clebsch-Gordan coefficients $C_{11L-1}^{L'0}$ and
$C_{1-1L1}^{L'0}$ coincide if $L'=L \pm 1$
(Ref.~\cite{Varshalovich}, section 8.4.3, formula (11)) and vanish
if $L' < L-1$ or $L' > L+1$ (Ref.~\cite{Varshalovich}, section
8.1.1, formula (1)), so the only contribution
to~\eqref{difference-T} makes $L'=L$, when $C_{11L-1}^{L0} -
C_{1-1L1}^{L0} = \sqrt{2}$ (Ref.~\cite{Varshalovich}, section
8.5.1, formula (8)). The Wigner 6$j$-symbol $\left\{
\begin{array}{ccc}
  1 & L & L \\
  j & j & j \\
\end{array} \right\} = \frac{1}{2} (-1)^{2j+L+1} \sqrt{\frac{L(L+1)}{j(j+1)(2j+1)(2L+1)}}$ (Ref.~\cite{Varshalovich}, section 9.5.4,
formula (21)). Finally, $J_+ T_{L-1}^{(j)} + J_- T_{L1}^{(j)} =
\sqrt{L(L+1)} T_{L0}^{(j)}$. Substituting this result in
formula~\eqref{Phi-on-TL0}, we conclude that $T_{L0}^{(j)}$ is the
eigenoperator of $\Phi$ and corresponds to the eigenvalue
$\lambda_L$ given by formula~\eqref{lambda-L}.

Due to the $SU(2)$-covariance of $\Phi$
(proposition~\ref{prop-SU2-cov}), the operators $U_g T_{L0}^{(j)}
U_g^{\dag}$ are eigenoperators of $\Phi$ too and correspond to the
eigenvalue $\lambda_L$. It is known that there are exactly $2L+1$
linear independent operators $U_g T_{L0}^{(j)} U_g^{\dag}$ if
$U_g$ is a representation of the $SU(2)$ group (see
Ref.~\cite{fm-2009}, formula (11), where $S_L^{(j)}$ is
proportional to the operator $T_{L0}^{(j)}$, and
Ref.~\cite{fm-2010}, text after formula (36)). Therefore,
eigenvalues $\lambda_L$ are $(2L+1)$-fold degenerate. Since
$\sum_{L=0}^{2j} (2L+1) = (2j+1)^2$, the given eigenvalues are the
only ones and constitute the spectrum of $\Phi$.
\end{proof}

In the latter proposition, for the case $L=1$ the generators
$J_x$, $J_y$, $J_z$ are exactly the three linear independent
eigenoperators of $\Phi: \mathcal{B}(\mathcal{H}_{2j+1}) \mapsto
\mathcal{B}(\mathcal{H}_{2j+1})$ corresponding to the eigenvalue
$1 - \frac{1}{j(j+1)}$.

It is not hard to see that if $L=2j$, then $\lambda_L < 0$.
Negativity of the eigenvalue implies that the Landau--Streater
channel cannot be obtained via positive divisible Hermitian
evolution~\cite{chruscinski-macchiavello-maniscalco-2017} for any
$j$.

\subsection{Spectrum of the output} \label{section-spectrum-output}

Let us consider spectral properties of the output operator, i.e.,
the spectrum of $\Phi[X]$, where $X$ is a Hermitian input
operator.

It is not hard to see that in the case $j=\frac{1}{2}$ the
spectrum of $\Phi[X]$ is $\{\frac{1}{3}(x_1+2x_2),\frac{1}{3}(2x_1
+ x_2)\}$ provided the spectrum of the input density operator
$\rho$ is $\{x_1,x_2\}$.

\begin{corollary} \label{corollary-qubit-entropy}
The output purity and entropy of the Landau--Streater channel
$\Phi: \mathcal{B}(\mathcal{H}_2) \mapsto
\mathcal{B}(\mathcal{H}_2)$ for all pure spin-$\frac{1}{2}$ input
states $| \psi \rangle \langle \psi |$ are equal to $\frac{5}{9}$
and $\log 3 - \frac{2}{3}$, respectively.
\end{corollary}

The case $j=1$ is more involved, but in this case the spectrum of
the output also depends only on the spectrum of the input, as we
show in the following proposition.

\begin{proposition} \label{prop-spectrum-output}
Suppose a Hermitian operator $X \in \mathcal{B}(\mathcal{H}_3)$
with spectrum $\{x_1,x_2,x_3\}$. The output operator $\Phi[X]$ of
the Landau--Streater channel $\Phi: \mathcal{B}(\mathcal{H}_3)
\mapsto \mathcal{B}(\mathcal{H}_3)$ has spectrum $\{
\frac{1}{2}(x_1+x_2),\frac{1}{2}(x_1+x_3),\frac{1}{2}(x_2+x_3)
\}$.
\end{proposition}

\begin{proof}
In the basis of eigenvectors of $J_z$, the action of the
Landau--Streater channel reads
\begin{equation}
\Phi \left[ \begin{pmatrix}
X_{11} & X_{12} & X_{13} \\
X_{21} & X_{22} & X_{23} \\
X_{31} & X_{32} & X_{33}
\end{pmatrix} \right] =
\frac{1}{2}
\begin{pmatrix}
X_{11} + X_{22} & X_{23} & -X_{13} \\
X_{32} & X_{11} + X_{33} & X_{12} \\
-X_{31} & X_{21} & X_{22} + X_{33}
\end{pmatrix}.
\end{equation}

Since the Landau--Streater channel $\Phi:
\mathcal{B}(\mathcal{H}_3) \mapsto \mathcal{B}(\mathcal{H}_3)$ is
globally unitarily covariant by proposition~\ref{prop-U3-cov}, the
spectrum of the output operator $\Phi[X]$ depends only on the
spectrum of the input operator $X$. To find the explicit relation
between the spectra we consider the unitary operator $U$ realizing
the transition from the basis of eigenvectors of $X$ to the basis
of eigenvectors of $J_z$. Then in the basis of eigenvectors of
$J_z$ we have $U X U^{\dag} = {\rm diag}(x_1,x_2,x_3)$ and $\Phi[U
X U^{\dag}] = {\rm diag} \left( \frac{1}{2} (x_1+x_2),\frac{1}{2}
(x_1+x_3),\frac{1}{2} (x_2+x_3) \right)$. Thanks to the global
unitary covariance, the latter diagonal matrix is exactly the
spectrum of $\Phi[X]$.
\end{proof}

The spectral property of the Landau--Streater channel $\Phi:
\mathcal{B}(\mathcal{H}_3) \mapsto \mathcal{B}(\mathcal{H}_3)$
resembles that of the depolarizing channel, but the
Landau--Streater channel is not depolarizing in the case $j = 1$.
This peculiarity is ascribed to the close relation between the
Landau--Streater channel and the Werner--Holevo channel $\Phi_{\rm
WH}: \mathcal{B}(\mathcal{H}_d) \mapsto
\mathcal{B}(\mathcal{H}_d)$ defined through transposition $\top$
in some orthonormal basis via formula $\Phi_{\rm WH}[X] =
\frac{1}{d-1} \left( {\rm tr}[X] I - X^{\top} \right)$,
Ref.~\cite{wh-2002}. It turns out that if $d=3$ and transposition
$\top$ is performed in the basis of eigenstates of $J_z$, then the
Landau--Streater channel $\Phi: \mathcal{B}(\mathcal{H}_3) \mapsto
\mathcal{B}(\mathcal{H}_3)$ is merely the Werner--Holevo channel
concatenated with a unitary channel:
\begin{equation} \label{LS-WH}
\Phi[X] = \Phi_{\rm WH}[W X W^{\dag}] = \frac{1}{2} \left( {\rm
tr}[X] I - W X^{\top} W^{\dag} \right),
\end{equation}

\noindent where $W = | 1,1 \rangle \langle 1,-1 | - | 1,0 \rangle
\langle 1,0 | + | 1,-1 \rangle \langle 1,1 |$ is a unitary
operator.

Since the spectrum of pure states consists of $1$ and zeros, we
can make conclusions about the output purity $\mu_{\rm out} = {\rm
tr}\left[ (\Phi[\rho])^2 \right]$ and the output entropy $S_{\rm
out} = - {\rm tr} \left[ \Phi[\rho] \log \Phi[\rho] \right]$ for
pure input states in cases $j=\frac{1}{2}$ and $j=1$. Hereafter,
$\log$ is understood as $\log_2$ if one measures the entropy and
capacity in bits.

\begin{corollary} \label{corollary-qutrit-entropy}
The output purity and entropy of the Landau--Streater channel
$\Phi: \mathcal{B}(\mathcal{H}_3) \mapsto
\mathcal{B}(\mathcal{H}_3)$ for all pure spin-$1$ input states $|
\psi \rangle \langle \psi |$ are equal to $\frac{1}{2}$ and $1$,
respectively.
\end{corollary}

In the case $j > 1$, the spectrum of $\Phi[X]$ depends not only on
the spectrum of $X$ but also on the particular form of the
operator $X$. For instance, in the case $j=\frac{3}{2}$, one can
consider two different pure input states
$|\frac{3}{2},\frac{3}{2}\rangle\langle \frac{3}{2},\frac{3}{2}|$
and $|\frac{3}{2},\frac{1}{2}\rangle\langle
\frac{3}{2},\frac{1}{2}|$ with identical spectra $\{1,0,0,0\}$. By
formula~\eqref{spec-phi-jm} ${\rm Spec} \left( \Phi \left[
|\frac{3}{2},\frac{3}{2}\rangle\langle \frac{3}{2},\frac{3}{2}|
\right] \right) = \{\frac{3}{5},\frac{2}{5},0,0\}$ and ${\rm Spec}
\left( \Phi \left[ |\frac{3}{2},\frac{1}{2}\rangle\langle
\frac{3}{2},\frac{1}{2}| \right] \right) =
\{\frac{2}{5},\frac{1}{4},\frac{8}{15},0\}$, so spectra of output
states may not coincide. Similarly to the case $j=\frac{3}{2}$,
for any $j>1$ one can always take pure input states
$|j,j\rangle\langle j,j|$ and $|j,j-1\rangle\langle j,j-1|$ and
make sure that ${\rm Spec} \left( \Phi \left[ |j,j\rangle\langle
j,j| \right] \right) \ne {\rm Spec} \left( \Phi \left[
|j,j-1\rangle\langle j,j-1| \right] \right)$.

\subsection{The maximal $p$-norm and the minimal output entropy} \label{section-multiplicativity}

The maximal $p$-norm of a channel $\Phi$ is defined by the formula
\begin{equation}
\nu_p(\Phi) = \sup_{\rho} \{ \| \Phi[\rho] \|_p \},
\end{equation}

\noindent where $\| \Phi[\rho] \|_p = \left\{ {\rm tr}
(\Phi[\rho])^p \right\}^{1/p}$ is the Schatten $p$-norm of
$\Phi[\rho]$. The maximal 2-norm is merely the square root of the
maximal output purity. Before we proceed to the analysis of
maximal $p$-norm and minimal output entropy of the
Landau--Streater channel, we prove auxiliary results following
from the theory of angular momentum.

\begin{lemma} \label{lemma-1}
Let ${\bf k}\in\mathbb{R}^3$ be a unit vector, $|{\bf k}| =
\sqrt{k_1^2+k_2^2+k_3^2} = 1$. The spectrum of operator
$\sum_{\alpha=1}^3 k_{\alpha} J_{\alpha}$ is $\{ m \}_{m=-j}^j$.
\end{lemma}
\begin{proof}
Physically, the operator $\sum_{\alpha=1}^3 k_{\alpha} J_{\alpha}$
is the spin projection operator onto axis ${\bf k}$ and therefore
has the same spectrum as any of operators $J_x$, $J_y$, $J_z$.
Mathematically, there exists a unitary operator $U_g: {\cal
B}({\cal H}_{2j+1}) \mapsto {\cal B}({\cal H}_{2j+1})$, $g\in
SU(2)$, such that $U_g^{\dag} J_z U_g = \sum_{\beta=1}^3 k_{\beta}
J_{\beta}$, cf. formula~\eqref{rotation} with $k_{\beta} =
Q_{3\beta}$, where $Q$ is orthogonal. Hence, ${\rm Spec}\left(
\sum_{\alpha=1}^3 k_{\alpha} J_{\alpha} \right) = {\rm Spec}(J_z)
= \{m\}_{m=-j}^j$.
\end{proof}

The eigenvector $| \psi_{\bf k} \rangle$ of operator
$\sum_{\alpha=1}^3 k_{\alpha} J_{\alpha}$ corresponding to the
maximal eigenvalue $j$ will be referred to as a vector with the
maximal spin polarization. Clearly, $| \psi_{\bf k} \rangle =
U_g^{\dag} | j,j \rangle$, where $U_g$ is the unitary operator
used in the proof of Lemma~\ref{lemma-1}.

\begin{lemma} \label{lemma-2}
Let ${\bf k}\in\mathbb{R}^3$, $|{\bf k}|=1$. The maximum of
expression $ \left\| \sum_{\alpha=1}^3 k_{\alpha} J_{\alpha} |
\psi \rangle \right\|^2$ with respect to normalized vectors $|
\psi \rangle \in {\cal H}_{2j+1}$ equals $j^2$ and is attained at
the state $| \psi_{\bf k} \rangle$ with the maximal spin
polarization.
\end{lemma}
\begin{proof}
By Lemma~\ref{lemma-1}, the spectrum of $\left( \sum_{\alpha=1}^3
k_{\alpha} J_{\alpha} \right)^2$ reads $\{m^2\}_{m=-j}^j$.
Therefore, $\| \left( \sum_{\alpha=1}^3 k_{\alpha} J_{\alpha}
\right) | \psi \rangle \|^2 = \langle \psi | \left(
\sum_{\alpha=1}^3 k_{\alpha} J_{\alpha} \right)^2 | \psi \rangle
\leq j^2 \langle \psi | \psi \rangle = j^2$ and \newline $\langle
\psi_{\bf k} | \left( \sum_{\alpha=1}^3 k_{\alpha} J_{\alpha}
\right)^2 | \psi_{\bf k} \rangle = j^2$.
\end{proof}

\begin{lemma} \label{lemma-3}
If $j \geq 1$, then $\langle \psi | J_z | \psi \rangle^2 \leq 9
j^2 \frac{j^2 - \langle \psi | J_x^2 | \psi \rangle}{2j-1}$ for
all normalized vectors $| \psi \rangle \in {\cal H}_{2j+1}$.
\end{lemma}
\begin{proof}
By Lemma~\ref{lemma-1}, the spectral decomposition of $J_x^2$
reads $J_x^2 = j^2 (P_{j} + P_{-j}) + \sum_{m=-j+1}^{j-1} m^2
P_m$, where $P_m = |j,m\rangle_x\langle j,m|$ and $J_x
|j,m\rangle_x = m |j,m\rangle_x$. The average value $\langle J_x^2
\rangle = j^2 - \epsilon \leq p j^2 + (1-p)(j-1)^2$, where $p =
\langle (P_{j} + P_{-j}) \rangle$. Therefore, $1-p \leq
\frac{\epsilon}{2j-1} = \frac{j^2 - \langle J_x^2 \rangle}{2j-1}$.

Let $|\psi\rangle = c_j|j,j\rangle_x + c_{-j}|j,-j\rangle_x +
\sum_{m=-j+1}^{j-1} c_m |j,m\rangle_x$. Note that ${}_x\langle
j,j| J_z |j,j\rangle_x = 0$, ${}_x\langle j,-j| J_z |j,-j\rangle_x
= 0$, $ {}_x\langle j,j| J_z |j,-j\rangle_x = 0$ if $j \geq 1$, $p
= |c_j|^2+|c_{-j}|^2$, and $1-p = \sum_{m=-j+1}^{j-1} |c_m|^2$. We
have
\begin{eqnarray}
\langle \psi | J_z | \psi \rangle & = & 2 {\rm Re} \left(
\sum_{m=-j+1}^{j-1} \overline{c_m} \, {}_x\langle j,m| \right) J_z
\left( c_j|j,j\rangle_x + c_{-j}|j,-j\rangle_x \right)
\nonumber\\
&& + \left(\sum_{m=-j+1}^{j-1} \overline{c_m} \, {}_x\langle j,m|
\right) J_z \left(\sum_{m'=-j+1}^{j-1} c_{m'} | j,m' \rangle_x
\right) \nonumber \\
& \leq & 2\sqrt{p}\sqrt{1-p}\, j + (1-p) j \leq 3 \sqrt{1-p} \, j.
\end{eqnarray}

Noticing that $-J_z$ has the same spectrum as $J_z$ and arguing as
above, we see that $\langle \psi | J_z | \psi \rangle \geq
2\sqrt{p}\sqrt{1-p}\, (-j) + (1-p) (-j) \geq - 3 \sqrt{1-p} \, j$.
Thus, $|\langle \psi | J_z | \psi \rangle| \leq 3 \sqrt{1-p} \,
j$. Squaring both sides of this inequality and recalling $1-p \leq
\frac{j^2 - \langle J_x^2 \rangle}{2j-1}$, we get the statement of
Proposition~ concludes the proof.
\end{proof}

Lemma~\ref{lemma-3} shows that the average value $\langle J_z
\rangle$ cannot be large when $\langle J_x^2 \rangle$ is close to
its maximal value $j^2$. Lemma~\ref{lemma-3} obviously remains
valid if one replaces $J_x$ by $J_y$.

\begin{lemma} \label{lemma-4}
Let the vectors ${\bf k},{\bf l} \in \mathbb{R}^3$ satisfy $|{\bf
k}|^2 + |{\bf l}|^2 = 1$, then $\left\| \sum_{\alpha=1}^3
(k_{\alpha} + i l_{\alpha}) J_{\alpha} | \psi \rangle \right\|^2
\leq \max(j,j^2)$ for all normalized vectors $| \psi \rangle \in
{\cal H}_{2j+1}$.
\end{lemma}
\begin{proof}
Suppose ${\bf k}$ and ${\bf l}$ are linearly dependent, i.e.,
${\bf k} = |{\bf k}| {\bf n}$ and ${\bf l} = |{\bf l}| {\bf n}$
for some unit vector ${\bf n} \in \mathbb{R}^3$. Then $\left\|
\sum_{\alpha=1}^3 (k_{\alpha} + i l_{\alpha}) J_{\alpha} | \psi
\rangle \right\|^2 = (|{\bf k}|^2 + |{\bf l}|^2) \left\|
\sum_{\alpha=1}^3 n_{\alpha} J_{\alpha} | \psi \rangle \right\|^2
\leq j^2$ by Lemma~\ref{lemma-2}.

Suppose ${\bf k}$ and ${\bf l}$ are linearly independent. Note
that $\left\| \sum_{\alpha=1}^3 (k_{\alpha} + i l_{\alpha})
J_{\alpha} | \psi \rangle \right\|^2 = \langle \psi | F | \psi
\rangle$, where
\begin{equation} \label{F-operator}
F:=\sum_{\alpha,\beta}(k_{\alpha} - i l_{\alpha})(k_{\beta} + i
l_{\beta}) J_{\alpha} J_{\beta} = \left( \sum_{\alpha} k_{\alpha}
J_{\alpha} \right)^2 + \left( \sum_{\alpha} l_{\alpha} J_{\alpha}
\right)^2 + \sum_{\gamma} [{\bf l} \times {\bf k}]_{\gamma}
J_{\gamma}.
\end{equation}

\noindent Here, we have used the commutation relation
$[J_{\alpha},J_{\beta}] = i \sum_{\gamma} e_{\alpha\beta\gamma}
J_{\gamma}$ and the notation $[{\bf l} \times {\bf k}]$ for the
conventional cross product of vectors ${\bf l}$ and ${\bf k}$. Let
the angle between vectors ${\bf k}$ and ${\bf l}$ be $\vartheta$.
Consider a rotation $Q$ in $\mathbb{R}^3$ such that $Q\left(
\frac{\bf k}{|{\bf k}|} + \frac{\bf l}{|{\bf l}|} \right)$ is
aligned with the positive direction of axis $x$ and $Q[{\bf l}
\times {\bf k}]$ is aligned with the positive direction of axis
$z$. In other words, the vectors $Q{\bf k}$ and $Q{\bf l}$ belong
to the $xy$-plane, and the axis $x$ is a bisector of the angle
between vectors $Q{\bf k}$ and $Q{\bf l}$. The vector $Q[{\bf l}
\times {\bf k}]$ is perpendicular to both $Q{\bf k}$ and $Q{\bf
l}$ and has length $|[{\bf l} \times {\bf k}]| = |{\bf k}| \,
|{\bf l}| \sin\vartheta$. Therefore, the vector $Q{\bf k}$ has
coordinates $(|{\bf k}|\cos\frac{\vartheta}{2},|{\bf
k}|\sin\frac{\vartheta}{2},0)$, the vector $Q{\bf l}$ has
coordinates $(|{\bf l}|\cos\frac{\vartheta}{2},-|{\bf
l}|\sin\frac{\vartheta}{2},0)$, and the vector $Q[{\bf l} \times
{\bf k}]$ has coordinates $(0,0,|{\bf k}| \, |{\bf l}|
\sin\vartheta)$. The corresponding unitary rotation $U_Q \in
\{U_g\}_{g\in SU(2)}$ transforms the spin operators in accordance
with formula~\eqref{rotation} as follows:
\begin{eqnarray}
\label{Jx-rotation} U_Q \left( \sum_{\alpha} k_{\alpha} J_{\alpha}
\right) U_Q^{\dag} &=& \sum_{\beta} (Q{\bf k})_{\beta} J_{\beta} =
|{\bf k}| \left( \cos\frac{\vartheta}{2} J_x +
\sin\frac{\vartheta}{2} J_y \right), \\
U_Q \left( \sum_{\alpha} l_{\alpha} J_{\alpha} \right) U_Q^{\dag}
&=& \sum_{\beta} (Q{\bf l})_{\beta} J_{\beta} = |{\bf l}| \left(
\cos\frac{\vartheta}{2} J_x -
\sin\frac{\vartheta}{2} J_y \right), \\
U_Q \left( \sum_{\gamma} [{\bf l} \times {\bf k}]_{\gamma}
J_{\gamma} \right) U_Q^{\dag} &=& \sum_{\beta} (Q[{\bf l} \times
{\bf k}])_{\beta} J_{\beta} = |{\bf k}| \, |{\bf l}| \sin\vartheta
J_z. \label{Jz-rotation}
\end{eqnarray}

\noindent Substituting \eqref{Jx-rotation}--\eqref{Jz-rotation} in
\eqref{F-operator} and taking into account that $|{\bf k}|^2+|{\bf
l}|^2 = 1$, we get
\begin{eqnarray}
U_Q F U_Q^{\dag} &=& \cos^2\frac{\vartheta}{2} J_x^2 +
\sin^2\frac{\vartheta}{2} J_y^2 + (|{\bf k}|^2 - |{\bf l}|^2)
\sin\frac{\vartheta}{2} \cos\frac{\vartheta}{2} \left( J_x J_y +
J_y J_x \right) + |{\bf
k}| \, |{\bf l}| \sin\vartheta J_z \nonumber\\
&=& \frac{1}{4} \cos\vartheta (J_+^2 + J_-^2) + \sin\vartheta
\left( \frac{1}{4i} (|{\bf k}|^2 - |{\bf l}|^2) (J_+^2 - J_-^2) +
|{\bf k}| \, |{\bf l}| J_z \right) + \frac{1}{2} (J_x^2 +
J_y^2).\nonumber\\
\end{eqnarray}

\noindent Quantities $|{\bf k}|^2 - |{\bf l}|^2$ and $2|{\bf k}|
\, |{\bf l}|$ can be treated as $\cos\eta$ and $\sin\eta$ for some
real $\eta$, respectively, because $(|{\bf k}|^2 - |{\bf l}|^2)^2
+ 4 |{\bf k}|^2 |{\bf l}|^2 = (|{\bf k}|^2 + |{\bf l}|^2)^2 = 1$.
Hence,
\begin{equation}
U_Q F U_Q^{\dag} = \frac{1}{4} \cos\vartheta (J_+^2 + J_-^2) +
\sin\vartheta \left( \frac{\cos\eta}{4i} (J_+^2 - J_-^2) +
\frac{\sin\eta}{2} J_z \right) + \frac{1}{2} (J_x^2 + J_y^2).
\end{equation}

If $j=1/2$, then $J_+^2 = J_-^2 = 0$, $J_x^2=J_y^2=\frac{1}{4}I$,
and $U_Q F U_Q^{\dag} = \frac{1}{2}\sin\vartheta\sin\eta J_z +
\frac{1}{4}I$. Clearly, $\langle \psi | U_Q F U_Q^{\dag} | \psi
\rangle \leq \frac{1}{2} = j$.

If $j=1$, then the matrix of operator $U_Q F U_Q^{\dag}$ in the
basis $\{|j,m\rangle\}_{m=-j}^j$ has a rather simple form. Its
eigenvalues do not depend on $\vartheta$ and $\eta$ and read
$1,1,0$.

For the cases $j=3/2$ and $j=2$ one can find eigenvalues of $U_Q F
U_Q^{\dag}$ and maximize them with respect to $\vartheta$ and
$\eta$ to get upper bounds $9/4$ and $4$, respectively.

For $j>2$ we develop the following technique.

Note that $A \cos\eta + B \sin\eta \leq \sqrt{A^2 + B^2}$ for
$A,B,\eta\in\mathbb{R}$, so the average value
\begin{equation}
\left\langle \left( \frac{1}{4i} (|{\bf k}|^2 - |{\bf l}|^2)
(J_+^2 - J_-^2) + |{\bf k}| \, |{\bf l}| J_z \right) \right\rangle
\leq \frac{1}{4} \sqrt{ \langle i(J_+^2 - J_-^2) \rangle^2 + 4
\langle J_z \rangle^2}.
\end{equation}

\noindent Similarly, $C\sin\vartheta + D\cos\vartheta \leq
\sqrt{C^2+D^2}$ for all $\vartheta,C,D \in \mathbb{R}$, therefore
\begin{equation} \label{average-inequality}
\langle U_Q F U_Q^{\dag} \rangle \leq \frac{1}{4} \sqrt{\langle
(J_+^2 + J_-^2) \rangle^2 + \langle i(J_+^2 - J_-^2) \rangle^2 + 4
\langle J_z \rangle^2} + \frac{1}{2} \langle (J_x^2 + J_y^2)
\rangle.
\end{equation}

Suppose the maximum in the right hand side of
\eqref{average-inequality} is attained at some vector
$|\psi_0\rangle$, then this maximum is also attained at the vector
$|\psi_{\theta} \rangle = e^{- i J_z \theta} |\psi_0\rangle$ due
to invariance of \eqref{average-inequality} with respect to
rotations around axis $z$. On the other hand, $\langle
\psi_{\theta} | J_+ | \psi_{\theta} \rangle = e^{i\theta} \langle
\psi_{0} | J_+ | \psi_{0} \rangle$, which means that $\langle J_+
\rangle$ can always be chosen to be real, so $\langle J_+ \rangle
= \langle J_- \rangle$. In other words,
\begin{eqnarray} \label{average-inequality}
\langle U_Q F U_Q^{\dag} \rangle & \leq & \max_{\psi: \langle \psi
| \psi \rangle = 1} \frac{1}{4} \sqrt{\langle \psi | (J_+^2 +
J_-^2) | \psi \rangle^2 + 4 \langle \psi | J_z | \psi \rangle^2} +
\frac{1}{2}
\langle \psi | (J_x^2 + J_y^2) | \psi \rangle \nonumber\\
& = & \max_{\psi: \langle \psi | \psi \rangle = 1} \frac{1}{2}
\sqrt{\langle \psi | (J_x^2 - J_y^2) | \psi \rangle^2 +  \langle
\psi | J_z | \psi \rangle^2} + \frac{1}{2} \langle \psi | (J_x^2 +
J_y^2) | \psi \rangle. \label{rhs-to-abc}
\end{eqnarray}

Denote $a = \langle \psi | J_x^2 | \psi \rangle$, $b = \langle
\psi | J_x^2 | \psi \rangle$, and $c = \langle \psi | J_z | \psi
\rangle^2$. The dispersion of spin projection onto axis $z$ denote
$d = \langle \psi | J_z^2 | \psi \rangle - \langle \psi | J_z |
\psi \rangle^2 = \langle \psi | J_z^2 | \psi \rangle - c$. Note
that $0 \leq a,b,c \leq j^2$ and $d \geq 0$. Since
$J_x^2+J_y^2+J_z^2 = j(j+1)I_{2j+1}$, we have $a+b+c+d = j(j+1)$.
Finally, from Lemma~\ref{lemma-3} it follows that $c \leq 9 j^2
\frac{j^2 - a}{2j-1}$ and $c \leq 9 j^2 \frac{j^2 - b}{2j-1}$.
Therefore, we simplify \eqref{rhs-to-abc} as follows:
\begin{equation} \label{abc}
\langle U_Q F U_Q^{\dag} \rangle \leq \frac{1}{2} \max_{ \small
\begin{array}{c}
  a,b,c,d: \\
  0 \leq a,b,c \leq j^2 \\
  0 \leq d \\
  (2j-1) c \leq 9
j^2 (j^2 - a) \\
  (2j-1) c \leq 9
j^2 (j^2 - b) \\
  a+b+c+d = j(j+1) \\
\end{array}} \sqrt{(a-b)^2+c} + a + b.
\end{equation}

Using the method of Lagrange multipliers one can readily see that
the maximum in the right hand side of \eqref{abc} is attained on
the boundary of region for parameters $a,b,c,d$. If $j>2$, then
the maximum equals $j^2$ and is attained at points with $a=j^2$
and $c=0$ or $b=j^2$ and $c=0$.

Although we have considered the cases $j=1,\frac{3}{2},2$ and
$j>2$ separately, their results can be unified, namely, $\langle
U_Q F U_Q^{\dag} \rangle \leq j^2$ if $j \geq 1$. Recalling the
fact $\langle U_Q F U_Q^{\dag} \rangle \leq \frac{1}{2}$ if
$j=\frac{1}{2}$, we obtain that $\langle U_Q F U_Q^{\dag} \rangle
\leq \max(j,j^2)$. Since this bound is valid for all normalized
states $|\psi\rangle$, we finally conclude that $\langle F \rangle
\leq \max(j,j^2)$.
\end{proof}

\begin{proposition} \label{prop-p-norm}
The maximal $p$-norm ($p \geq 1$) and the minimal output entropy
of the Landau--Streater channel $\Phi:
\mathcal{B}(\mathcal{H}_{2j+1}) \mapsto
\mathcal{B}(\mathcal{H}_{2j+1})$ are equal to
\begin{equation} \label{p-norm-Phi}
\nu_p(\Phi) = \frac{(j^p + 1)^{1/p}}{j+1} \quad \text{and} \quad
S_{\min}(\Phi) = \log (j+1) - \frac{j}{j+1} \log j,
\end{equation}

\noindent respectively, and are attained at the state
$|j,j\rangle$.
\end{proposition}
\begin{proof}
Let $|\psi\rangle \langle \psi |$ be a pure state at which the
maximal $\infty$-norm is attained. Then $\| \Phi[|\psi\rangle
\langle \psi |] \|_{\infty} = \lambda$, where $\lambda$ is the
maximal output eigenvalue. On the other hand,
\begin{equation} \label{infinite-norm}
\lambda = \max_{\chi \neq 0} \frac{\langle \chi | \,
\Phi[|\psi\rangle \langle \psi |] \, | \chi \rangle} {\langle \chi
| \chi \rangle} = \max_{\chi \neq 0} \frac{ \sum_{\alpha=1}^3
 | \langle \varphi_{\alpha} | \chi \rangle |^2 }{j(j+1) \langle \chi | \chi \rangle}, \qquad | \varphi_{\alpha} \rangle = J_{\alpha} | \psi \rangle.
\end{equation}

The vector $|\chi\rangle$ maximizing~\eqref{infinite-norm} must
belong to a linear span of vectors
$\{|\varphi_{\alpha}\rangle\}_{\alpha=1}^3$, i.e. $|\chi\rangle =
\sum_{\beta=1}^3 c_{\beta} |\varphi_{\beta}\rangle$. Introduce the
Hermitian Gram matrix $G_{\alpha\beta} = \langle \varphi_{\alpha}
| \varphi_{\beta}
\rangle = \langle \psi | J_{\alpha} J_{\beta} | \psi \rangle$ and the vector $|c\rangle = \left(%
\begin{array}{ccc}
  c_1 & c_2 & c_3 \\
\end{array}%
\right)^{\top} \in \mathcal{H}_3$, then $\sum_{\alpha=1}^3
 | \langle \varphi_{\alpha} | \chi \rangle |^2 = \sum_{\alpha=1}^3
\left\vert \sum_{\beta=1}^3 G_{\alpha\beta} c_{\beta}
\right\vert^2 = \langle c | G^2 | c \rangle$ and $\langle \chi |
\chi \rangle = \langle c | G | c \rangle$.
Equation~\eqref{infinite-norm} can be further simplified with the
use of vector $| c ' \rangle = \sqrt{G} | c \rangle$:
\begin{equation} \label{lambda-throgh-c}
\lambda = \frac{1}{j(j+1)} \max_{ c: \ \sqrt{G}|c\rangle \neq 0 }
\frac{\langle c | G^2 | c \rangle}{\langle c | G | c \rangle} =
\frac{1}{j(j+1)} \max_{ c' \neq 0 } \frac{\langle c' | G | c'
\rangle}{\langle c' | c' \rangle} = \frac{1}{j(j+1)}
\|G\|_{\infty}.
\end{equation}

On the other hand, the $\infty$-norm of $G$ reads
\begin{equation} \label{G-norm}
\|G\|_{\infty} = \max_{u: \, u^{\dag}u = I} \sum_{\alpha,\beta =
1}^3 \overline{u_{\alpha 1}} G_{\alpha \beta} u_{\beta 1} =
\max_{u: \, u^{\dag}u = I} \left\| \sum_{\alpha=1}^3 u_{\alpha 1}
J_{\alpha} |\psi\rangle \right\|^2.
\end{equation}

Since $u$ is a unitary matrix, the vector ${\bf u} = \left(%
\begin{array}{ccc}
  u_{11} & u_{21} & u_{31} \\
\end{array}%
\right)^T = {\bf k} + i {\bf l}$, where ${\bf k},{\bf l} \in
\mathbb{R}^3$ and ${\bf u}^{\dag}{\bf u} = |{\bf k}|^2 + |{\bf
l}|^2 = 1$. By Lemma~\ref{lemma-4} for any $|\psi\rangle$ the
maximum in the right hand side of \eqref{G-norm} does not exceed
$\max(j,j^2)$. Therefore,
\begin{equation}
\lambda = \frac{1}{j(j+1)} \|G\|_{\infty} \leq \max \left(
\frac{j}{j+1},\frac{1}{j+1} \right).
\end{equation}

On the other hand, if $|\psi\rangle = |j,j\rangle$, then by
formula~\eqref{spec-phi-jm} we have $Spec \left( \Phi[| j,j
\rangle \langle j,j |] \right) = \left\{ \frac{j}{j+1},
\frac{1}{j+1}, 0, \ldots \right\}$. This implies that $\lambda =
\max(\frac{j}{j+1},\frac{1}{j+1})$ and $\| \Phi[|\psi\rangle
\langle \psi |] \|_{\infty}$ is attained at the vector
$|j,j\rangle$.

Denote $\boldsymbol{\lambda} = \left\{ \frac{j}{j+1},
\frac{1}{j+1}, 0, \ldots \right\}$. Since $\boldsymbol{\lambda}$
has only two nonzero components and the largest component is
$\nu_{\infty}(\Phi)$, then $\boldsymbol{\lambda}$ majorizes all
other output spectra $\boldsymbol{\mu}$. Here, we use the
conventional definition of majorization (Ref.~\cite{Tong},
Definition 12.1): a sequence of real numbers
$\boldsymbol{\lambda}=(\lambda_1,\lambda_2,\ldots,\lambda_n)$
majorizes another sequence of real numbers
$\boldsymbol{\mu}=(\mu_1,\mu_2,\ldots,\mu_n)$ if, after possible
renumeration, the terms of the sequences $\boldsymbol{\lambda}$
and $\boldsymbol{\mu}$ satisfy conditions $\lambda_1 \geq
\lambda_2 \geq \ldots \geq \lambda_n$, $\mu_1 \geq \mu_2 \geq
\ldots \geq \mu_n$, $\lambda_1 + \lambda_2 + \ldots + \lambda_k
\geq \mu_1 + \mu_2 + \ldots + \mu_k$ for each $k$, $1 \leq k \leq
n-1$, and $\lambda_1 + \lambda_2 + \ldots + \lambda_n = \mu_1 +
\mu_2 + \ldots + \mu_n$. In our case, $\lambda_1 =
\nu_{\infty}(\Phi)$, which guarantees $\lambda_1 \geq \mu_1$ for
all other output spectra $\boldsymbol{\mu}$. Moreover, $\lambda_1
+ \lambda_2 = 1 = \mu_1 + \mu_2 + \ldots + \mu_{2j+1} \geq \mu_1 +
\mu_2$, therefore $\boldsymbol{\lambda}$ majorizes any other
output spectrum $\boldsymbol{\mu}$. Since functions $y(x)=x{\rm
log}x$ and $y(x)=x^p$, $p\geq 1$, are convex, the Shannon entropy
$H({\bf x})=- \sum_{k=1}^n x_k {\rm log} x_k$ is a Schur-concave
function of ${\bf x} \in [0,1]^n$ and the output $p$-norm ${\cal
V}_p({\bf x}) = \left( \sum_{k=1}^{n} x_k^p \right)^{1/p}$ is a
Schur-convex function of ${\bf x} \in [0,1]^n$ (Ref.~\cite{Tong},
Definition 12.23, Theorem 12.27). Therefore, $H(\boldsymbol{\mu})
\geq H(\boldsymbol{\lambda})$ and ${\cal V}_p(\boldsymbol{\mu})
\leq {\cal V}_p(\boldsymbol{\lambda})$ for all output spectra
$\boldsymbol{\mu}$. This observation results in
formulas~\eqref{p-norm-Phi}.
\end{proof}

Corollaries~\ref{corollary-qubit-entropy}
and~\ref{corollary-qutrit-entropy} are merely consequences of
proposition~\ref{prop-p-norm}.

Consider the second tensor power $\Phi^{\otimes 2}$ of a channel
$\Phi: \mathcal{B}(\mathcal{H}) \mapsto \mathcal{B}(\mathcal{H})$.
Suppose a density operator $\rho \in \mathcal{B}(\mathcal{H})$.
Then for the factorized input $\rho^{\otimes 2}$ we have
$\Phi^{\otimes 2}[\rho^{\otimes 2}] = \left( \Phi[\rho]
\right)^{\otimes 2}$. Obviously, the purity ${\rm tr}\left[ \left(
\Phi^{\otimes 2}[\rho^{\otimes 2}] \right)^2 \right]$ of the state
$\Phi^{\otimes 2}[\rho^{\otimes 2}]$ is equal to the square of the
purity ${\rm tr}\left[ \left( \Phi[\rho] \right)^2 \right]$ of the
state $\Phi[\rho]$. However, if one uses \textit{entangled} input
states $\varrho_{\rm ent} \in \mathcal{B}(\mathcal{H}^{\otimes
2})$, then in general the purity of the state $\Phi^{\otimes
2}[\varrho_{\rm ent}]$ can be greater than the purity of all
possible factorized states $\Phi^{\otimes 2}[\rho^{\otimes 2}]$.
Therefore, in general $\nu_2(\Phi^{\otimes 2}) \geqslant \left(
\nu_2(\Phi) \right)^2$. Clearly, if $\nu_2(\Phi^{\otimes 2}) >
\left( \nu_2(\Phi) \right)^2$, then the maximal 2-norm for the
channel $\Phi^{\otimes 2}$ is attained at some entangled state.

Nevertheless, there exist some channels, for which
$\nu_2(\Phi^{\otimes 2}) = \left( \nu_2(\Phi) \right)^2$ and the
use of entangled inputs does not help to increase the output
purity~\cite{michalakis-2007}. Among such channels, there is a
class of unital qubit channels~\cite{king-2002}, so for the
Landau--Streater channel $\Phi: \mathcal{B}(\mathcal{H}_2) \mapsto
\mathcal{B}(\mathcal{H}_2)$ the multiplicativity of the maximal
2-norm holds. Since the Landau--Streater channel $\Phi:
\mathcal{B}(\mathcal{H}_3) \mapsto \mathcal{B}(\mathcal{H}_3)$
reduces to the Werner--Holevo channel, then the multiplicativity
of the maximal $p$-norm for such a channel holds for all $1
\leqslant p \leqslant 2$, Ref.~\cite{datta-2004}, and is violated
for $p
> 4.79$, Ref.~\cite{wh-2002}. The Landau--Streater channel for $j > 1$ cannot be analyzed in the
same way as the case $j=1$ and does not satisfy the known
sufficient criteria of multiplicativity of the maximal
$2$-norm~\cite{michalakis-2007}. Despite this fact, if $p=2$, our
numerical investigations of the cases $j = \frac{3}{2}$ and $j =
2$ show that the maximal $2$-norm is multiplicative within the
accuracy of calculations. We can make a conjecture that the
maximal $2$-norm is multiplicative for all Landau--Steater
channels $\Phi: \mathcal{B}(\mathcal{H}_{2j+1}) \mapsto
\mathcal{B}(\mathcal{H}_{2j+1})$.

\section{Complementary channel} \label{section-complementary}

According to the Stinespring's dilation theorem, the dual channel
$\Phi^{\dag}: \mathcal{B}(\mathcal{H}) \mapsto
\mathcal{B}(\mathcal{H})$ adopts a representation
(Ref.~\cite{holevo-2012}, Theorem 6.9)
\begin{equation}
\Phi^{\dag} [X] = V^{\dag} (X \otimes I_{\mathcal{K}}) V,
\end{equation}

\noindent where $I_{\mathcal{K}}$ is the identity operator in some
Hilbert space $\mathcal{K}$, $V$: $\mathcal{H} \mapsto \mathcal{H}
\otimes \mathcal{K}$ is an isometry operator, i.e., $V^{\dag}V =
I$.

In the case of the Landau--Streater channel $\Phi:
\mathcal{B}(\mathcal{H}_{2j+1}) \mapsto
\mathcal{B}(\mathcal{H}_{2j+1})$ the dual channel $\Phi^{\dag}$
coincides with $\Phi$ and the corresponding Stinespring's dilation
is achieved with the help of the isometry operator $V$:
$\mathcal{H}_{2j+1} \mapsto \mathcal{H}_{2j+1} \otimes
\mathcal{H}_3$ of the form
\begin{equation}
V = \frac{1}{\sqrt{j(j+1)}} \left(%
\begin{array}{c}
  J_x \\
  J_y \\
  J_z \\
\end{array}%
\right).
\end{equation}

Therefore, ${\rm dim}\mathcal{K} = 3$ and the Landau--Streater
channel $\Phi: \mathcal{B}(\mathcal{H}_{2j+1}) \mapsto
\mathcal{B}(\mathcal{H}_{2j+1})$ can be realized via a
3-dimensional environment. In the Schr\"{o}dinger picture of
system-environment interaction (Ref.~\cite{holevo-2012}, Theorem
6.9), we have
\begin{equation} \label{channel-sch}
\Phi[\rho] = {\rm tr}_{\mathcal{K}} \left[ V \rho V^{\dag} \right]
= {\rm tr}_{\mathcal{K}} \left[ U (\rho \otimes \xi) U^{\dag}
\right],
\end{equation}

\noindent where $\xi \in \mathcal{B}(\mathcal{H}_3)$ is the pure
initial environment state, $U: \mathcal{H}_{2j+1} \otimes
\mathcal{H}_3 \mapsto \mathcal{H}_{2j+1} \otimes \mathcal{H}_3$ is
the unitary evolution operator. The general technique of finding
$U$ is described, e.g., in Ref.~\cite{nielsen-2000}, section
8.2.3.

If one replaces the partial trace over environment ${\rm
tr}_{\mathcal{K}}$ by the partial trace over system ${\rm
tr}_{\mathcal{H}}$ in formula~\eqref{channel-sch}, then one
obtains a so-called complementary channel~\cite{holevo-2005}
$\widetilde{\Phi}: \mathcal{B}(\mathcal{H}) \mapsto
\mathcal{B}(\mathcal{K})$ (also referred to as conjugate
channel~\cite{king-2007}):
\begin{equation} \label{channel-complementary}
\widetilde{\Phi}[\rho] = {\rm tr}_{\mathcal{H}} \left[ V \rho
V^{\dag} \right] = {\rm tr}_{\mathcal{H}} \left[ U (\rho \otimes
\xi) U^{\dag} \right].
\end{equation}

\noindent In the case of the Landau--Streater channel $\Phi:
\mathcal{B}(\mathcal{H}_{2j+1}) \mapsto
\mathcal{B}(\mathcal{H}_{2j+1})$, the complementary channel $\Phi:
\mathcal{B}(\mathcal{H}_{2j+1}) \mapsto
\mathcal{B}(\mathcal{H}_{3})$ maps spin-$j$ states into
3-dimensional environment states (also known as qutrit states). In
what follows, we use the notation $I_d$ to denote the identity
operator $I: \mathcal{H}_d \mapsto \mathcal{H}_d$.

\begin{proposition} \label{prop-complementary-on-mixed}
The channel $\widetilde{\Phi}: \mathcal{B}(\mathcal{H}_{2j+1})
\mapsto \mathcal{B}(\mathcal{H}_{3})$, which is complementary to
the Landau--Streater channel $\Phi:
\mathcal{B}(\mathcal{H}_{2j+1}) \mapsto
\mathcal{B}(\mathcal{H}_{2j+1})$, transforms the maximally mixed
input state $\frac{1}{2j+1} I_{2j+1}$ into the maximally mixed
output state $\frac{1}{3} I_3$.
\end{proposition}
\begin{proof}
Denote by $V_{\alpha} = [j(j+1)]^{-1/2} J_{\alpha}$, $\alpha =
1,2,3$, the Kraus operators of $\Phi$ and by $\widetilde{V}_i$,
$i=1,\ldots,2j+1$, the Kraus operators of $\widetilde{\Phi}$.
These Kraus operators are mutually related with each other by
formula (Ref.~\cite{holevo-2005}, formula (12))
\begin{equation} \label{kraus-kraus}
\langle \alpha_{\mathcal{K}} | \widetilde{V}_i = \langle
i_{\mathcal{H}}| V_{\alpha},
\end{equation}

\noindent where $\{ i_{\mathcal{H}} \}_{i=1}^{2j+1}$ is the
orthonormal basis in $\mathcal{H}$ (input) and $\{
\alpha_{\mathcal{K}} \}_{\alpha = 1}^{3}$ is the orthonormal basis
in $\mathcal{K}$ (output). Multiplying \eqref{kraus-kraus} from
the left by $| \alpha_{\mathcal{K}} \rangle$ and summing over
$\alpha$, we get
\begin{equation} \label{kraus-through-kraus}
\widetilde{V}_i = \sum_{\alpha = 1}^{3} | \alpha_{\mathcal{K}}
\rangle \langle \alpha_{\mathcal{K}} | \widetilde{V}_i =
\sum_{\alpha = 1}^{3} | \alpha_{\mathcal{K}} \rangle \langle
i_{\mathcal{H}} | V_{\alpha} = \frac{1}{\sqrt{j(j+1)}}
\sum_{\alpha = 1}^{3} | \alpha_{\mathcal{K}} \rangle \langle
i_{\mathcal{H}} | J_{\alpha}.
\end{equation}

Action of the complementary channel $\widetilde{\Phi}$ on the
maximally mixed state $(2j+1)^{-1} I_{2j+1}$ reads
\begin{eqnarray}
\frac{1}{2j+1} \widetilde{\Phi}[I_{2j+1}] &=& \frac{1}{2j+1}
\sum_{i=1}^{2j+1} \widetilde{V}_i \widetilde{V}_i^{\dag}
\nonumber\\
&=& \frac{1}{j(j+1)(2j+1)} \sum_{i=1}^{2j+1}
\sum_{\alpha,\beta=1}^3 | \alpha_{\mathcal{K}} \rangle \langle
i_{\mathcal{H}} | J_{\alpha} J_{\beta}^{\dag}
|i_{\mathcal{H}}\rangle \langle \beta_{\mathcal{K}} | \nonumber\\
&=& \frac{1}{j(j+1)(2j+1)} \sum_{\alpha,\beta=1}^3  {\rm tr}
\left[ J_{\alpha} J_{\beta}^{\dag} \right] | \alpha_{\mathcal{K}}
\rangle \langle \beta_{\mathcal{K}} |.
\end{eqnarray}

\noindent Since $SU(2)$ generators $J_{\alpha}$, $\alpha=1,2,3$,
are Hermitian and satisfy the relation ${\rm tr} \left[ J_{\alpha}
J_{\beta} \right] = \frac{1}{3} j(j+1)(2j+1) \delta_{\alpha\beta}$
(Ref.~\cite{Varshalovich}, section 2.3.4, formula (11)), then
$\widetilde{\Phi}[\frac{1}{2j+1} I_{2j+1}] = \frac{1}{3} I_3$.
\end{proof}

A channel $\Phi$ is called degradable~\cite{Cubitt} if there
exists a channel $T$ such that $\widetilde{\Phi} = T \circ \Phi$.
Conversely, a channel $\Phi$ is called
antidegradable~\cite{Cubitt} if there exists a channel $T'$ such
that $\Phi = T' \circ \widetilde{\Phi}$. The structure of
degradable and antidegradable quantum channels is studied in
Ref.~\cite{Cubitt}. Further, we explore degradability and
antidegradability of the Landau--Streater channel for various
values of $j$.

\subsection{The case $j=1/2$} \label{subsection-complementary-1/2}

If $j=\frac{1}{2}$, then the Landau--Streater channel $\Phi$ is
antidegradable but not degradable. In fact, in this case $\Phi$ is
a qubit depolarization channel with depolarization parameter
$-\frac{1}{3}$, so it is entanglement breaking~\cite{ruskai-2003}
and, consequently, antidegradable. The Kraus operators of the
complementary channel $\widetilde{\Phi}$ are calculated via
formula~\eqref{kraus-through-kraus}, which results in the
following form of $\widetilde{\Phi}$:
\begin{equation}
\widetilde{\Phi}[X] = \frac{1}{3} \left(\left(
\begin{array}{cc}
 0 & 1 \\
 0 & -i \\
 1 & 0 \\
\end{array}
\right) X \left(
\begin{array}{ccc}
 0 & 0 & 1 \\
 1 & i & 0 \\
\end{array}
\right)+\left(
\begin{array}{cc}
 1 & 0 \\
 i & 0 \\
 0 & -1 \\
\end{array}
\right) X \left(
\begin{array}{ccc}
 1 & -i & 0 \\
 0 & 0 & -1 \\
\end{array}
\right) \right).
\end{equation}

\noindent The factoring map $T = \widetilde{\Phi} \circ \Phi^{-1}$
is well defined, and its normalized Choi matrix~\cite{Choi}
$\Omega_{T} = T \otimes {\rm Id}_{2} [ | \psi_+ \rangle \langle
\psi_+ | ]$, $| \psi_+ \rangle = \frac{1}{\sqrt{2}} \sum_{i=1}^2
|i\rangle \otimes |i\rangle$, reads
\begin{equation}
\Omega_{T} = \frac{1}{6} \left(
\begin{array}{cccccc}
 1 & 0 & 3 i & 0 & 0 & 3 \\
 0 & 1 & 0 & -3 i & -3 & 0 \\
 -3 i & 0 & 1 & 0 & 0 & 3 i \\
 0 & 3 i & 0 & 1 & 3 i & 0 \\
 0 & -3 & 0 & -3 i & 1 & 0 \\
 3 & 0 & -3 i & 0 & 0 & 1 \\
\end{array}
\right).
\end{equation}

\noindent Since $\Omega_{T}$ has negative eigenvalues, $T$ is not
completely positive, and $\Phi$ is not degradable.

\subsection{The case $j=1$} \label{subsection-complementary-1}

If $j=1$, then the Landau--Streater channel $\Phi$ is both
degradable and antidegradable. This follows from the fact that
$\Phi: \mathcal{B}(\mathcal{H}_{3}) \mapsto
\mathcal{B}(\mathcal{H}_{3})$ is unitarily equivalent to the
Werner--Holevo channel, which is both degradable and
antidegradable (Ref.~\cite{Cubitt}, section 2.2). For the sake of
completeness, we list the Kraus operators of the complementary
channel in this case:
\begin{equation}
\widetilde{V}_1 =
\begin{pmatrix}
0& \frac{1}{2} &0 \\
0& -\frac{ i }{2} &0 \\
\frac{1}{\sqrt{2}} &0&0
\end{pmatrix}, \quad
\widetilde{V}_2 = \begin{pmatrix}
\frac{1}{2} & 0 & \frac{1}{2} \\
\frac{ i }{2} & 0 & -\frac{ i }{2} \\
0&0&0
\end{pmatrix}, \quad
\widetilde{V}_3 = \begin{pmatrix}
0& \frac{1}{2} & 0 \\
0 & \frac{ i }{2} & 0 \\
0 & 0 & -\frac{1}{\sqrt{2}}
\end{pmatrix}.
\end{equation}

\subsection{The case $j \geq 3/2$} \label{subsection-complementary-3/2}

\begin{proposition} \label{prop-not-antidegradable}
The Landau--Streater channel $\Phi:
\mathcal{B}(\mathcal{H}_{2j+1}) \mapsto
\mathcal{B}(\mathcal{H}_{2j+1})$ is not antidegradable if $j
\geqslant \frac{3}{2}$.
\end{proposition}
\begin{proof}
Since doubly complementary channel $\widetilde{\widetilde{\Phi}}$
is unitarily equivalent to $\Phi$ (Ref.~\cite{holevo-2012},
Exercise 6.29), it is enough to demonstrate that
$\widetilde{\Phi}: \mathcal{B}(\mathcal{H}_{2j+1}) \mapsto
\mathcal{B}(\mathcal{H}_{3})$ is not degradable. The output space
for the complementary channel is 3-dimensional, so we use
Theorem~10 in Ref.~\cite{Cubitt}, which states that if
$\widetilde{\Phi}: \mathcal{B}(\mathcal{H}_{d}) \mapsto
\mathcal{B}(\mathcal{H}_{3})$ is degradable, then the Choi rank of
$\widetilde{\Phi}$ is at most 3. Choi rank is defined as the rank
of the Choi matrix~\cite{Choi} $\Omega_{\widetilde{\Phi}} =
\widetilde{\Phi} \otimes {\rm Id}_{d} [ | \psi_+ \rangle \langle
\psi_+ | ]$, with $| \psi_+ \rangle = \frac{1}{\sqrt{d}}
\sum_{i=1}^d |i\rangle \otimes |i\rangle$ being the maximally
entangled state. The Choi matrix of the complementary channel
$\widetilde{\Phi}: \mathcal{B}(\mathcal{H}_{2j+1}) \mapsto
\mathcal{B}(\mathcal{H}_{3})$ reads
\begin{eqnarray}
\Omega_{\widetilde{\Phi}} &=& \frac{1}{2j+1} \sum_{m,m'=-j}^j
\widetilde{\Phi} [|j,m\rangle \langle j,m'|] \otimes |j,m\rangle
\langle j,m'| \nonumber\\
&=& \frac{1}{j(j+1)(2j+1)} \sum_{\alpha,\beta=1}^3 \,
\sum_{m,m',i=-j}^j  \langle j,m' | J_{\beta} | i \rangle \langle i
| J_{\alpha} | j,m \rangle \, | \alpha \rangle \langle \beta |
\otimes | j,m \rangle \langle j,m' | \nonumber\\
&=& \frac{1}{j(j+1)(2j+1)} \sum_{\alpha,\beta=1}^3 | \alpha
\rangle \langle \beta | \otimes J_{\alpha}^{\top}
J_{\beta}^{\top} \nonumber\\
&=& \frac{1}{j(j+1)(2j+1)} \left(%
\begin{array}{ccc}
  J_x^2 & -J_x J_y & J_x J_z \\
  - J_y J_x & J_y^2 & -J_y J_z \\
  J_z J_x & - J_z J_y & J_z^2 \\
\end{array}%
\right) \nonumber\\
&=& \frac{1}{j(j+1)(2j+1)} \left(%
\begin{array}{c}
  J_x \\
  -J_y \\
  J_z \\
\end{array}%
\right)
\left(%
\begin{array}{ccc}
  J_x & -J_y & J_z \\
\end{array}%
\right), \label{sandwich}
\end{eqnarray}

\noindent where transposition is performed in the basis
$\{|j,m\rangle\}_{m=-j}^j$ and, therefore, $J_x^{\top} = J_x$,
$J_y^{\top} = - J_y$, $J_z^{\top} = J_z$. Denote $C = \left(%
\begin{array}{ccc}
  J_x & -J_y & J_z \\
\end{array}%
\right)$, then $j(j+1)(2j+1) \Omega_{\widetilde{\Phi}} =
C^{\dag}C$ and ${\rm rank}\, \Omega_{\widetilde{\Phi}} \leq {\rm
rank}\,C \leq 2j+1$. On the other hand, $CC^{\dag} =
j(j+1)I_{2j+1}$ and $j(j+1)(2j+1) C \Omega_{\widetilde{\Phi}}
C^{\dag} = C C^{\dag} C C^{\dag} = j^2(j+1)^2 I_{2j+1}$. This
implies that $2j+1 = {\rm rank}\, C \Omega_{\widetilde{\Phi}}
C^{\dag} \leq {\rm rank}\, \Omega_{\widetilde{\Phi}}$. Combining
both inequalities, we get ${\rm rank}\, \Omega_{\widetilde{\Phi}}
= 2j+1$. Thus, ${\rm rank}\, \Omega_{\widetilde{\Phi}} \geq 4$ if
$j \geqslant \frac{3}{2}$. Therefore, $\widetilde{\Phi}$ is not
degradable~\cite{Cubitt} and $\Phi$ is not antidegradable.
\end{proof}

\begin{proposition}
The Landau--Streater channel $\Phi:
\mathcal{B}(\mathcal{H}_{2j+1}) \mapsto
\mathcal{B}(\mathcal{H}_{2j+1})$ is not degradable if $j \geqslant
\frac{3}{2}$.
\end{proposition}
\begin{proof}
To prove that the factoring map $T = \widetilde{\Phi} \circ
\Phi^{-1}$ is not completely positive, it suffices to verify
negativity of some diagonal element of the Choi matrix $\Omega_T =
T \otimes {\rm Id}_{2j+1} [ | \psi_+ \rangle \langle \psi_+ | ]$,
$| \psi_+ \rangle = \frac{1}{\sqrt{2j+1}} \sum_{m'=-j}^{j}
|j,m'\rangle \otimes |j,m'\rangle$. Consider the diagonal element
\begin{eqnarray}
&& \langle \alpha | \otimes \langle j,m | \, \Omega_T \, | \alpha
\rangle \otimes | j,m \rangle = \frac{1}{2j+1} \langle \alpha |
\widetilde{\Phi} \circ \Phi^{-1} [| j,m \rangle \langle j,m |] |
\alpha \rangle
\nonumber\\
&& = \frac{1}{2j+1} \sum_{i=1}^{2j+1} \langle \alpha |
\widetilde{V}_i \, \Phi^{-1} [| j,m \rangle \langle j,m |] \,
\widetilde{V}_i^{\dag} | \alpha \rangle = \frac{1}{2j+1}
\sum_{i=1}^{2j+1} \langle i | V_{\alpha} \, \Phi^{-1} [| j,m
\rangle \langle j,m |] \, V_{\alpha}^{\dag} | i \rangle \nonumber\\
&& = \frac{1}{j(j+1)(2j+1)} {\rm tr} \left( J_{\alpha} \,
\Phi^{-1} [| j,m \rangle \langle j,m |] \, J_{\alpha} \right) =
\frac{1}{j(j+1)(2j+1)} {\rm tr} \left(J_{\alpha}^2 \, \Phi^{-1} [|
j,m
\rangle \langle j,m |] \right) \nonumber\\
&& = \frac{1}{j(j+1)(2j+1)} {\rm tr} \left( | j,m \rangle \langle
j,m | \, \Phi^{-1} [ J_{\alpha}^2 ] \right)
\label{Phi-selfdual-property} = \frac{1}{j(j+1)(2j+1)} \langle j,m
| \, \Phi^{-1} [ J_{\alpha}^2 ] | j,m \rangle.
\end{eqnarray}

\noindent In derivation of formula \eqref{Phi-selfdual-property}
we have taken into account that $\Phi$ is a self-dual map and
${\rm tr}(X \Phi^{-1}[Y]) = {\rm tr}
\left(\Phi\left[\Phi^{-1}[X]\right] \Phi^{-1}[Y] \right) = {\rm
tr} \left( \Phi^{-1}[X] \Phi^{\dag}\left[ \Phi^{-1}[Y] \right]
\right) = {\rm tr}(\Phi^{-1}[X] Y)$.

Let us fix $\alpha = 3$ and calculate the operator $\Phi[J_z^2]$
by using formula~\eqref{Phi-on-jm} and the spectral decomposition
$J_z = \sum_{m'=-j}^{j} m' | j,m' \rangle \langle j,m' |$:
\begin{eqnarray}
\Phi[J_z^2] &=& \sum_{m'=-j}^{j} (m')^2 \Phi[| j,m' \rangle
\langle j,m' |] = \frac{1}{j(j+1)} \sum_{m'=-j}^{j} (m')^2  \bigg[
(m')^2 |j,m'\rangle\langle j,m'|  \nonumber\\
&& \qquad + \frac{1}{2}\left( j(j+1)-m'(m'-1) \right)
|j,m'-1\rangle\langle j,m'-1| \nonumber\\
&& \qquad + \frac{1}{2}\left( j(j+1)-m'(m'+1)
\right) |j,m'+1\rangle\langle j,m'+1| \bigg] \nonumber\\
&=& \frac{1}{j(j+1)} \sum_{m'=-j}^{j} \bigg[
\big(j(j+1)-3\big)(m')^2 + j(j+1) \bigg] |j,m'\rangle\langle j,m'| \nonumber\\
&=& \frac{j(j+1)-3}{j(j+1)} J_z^2 + I_{2j+1}.
\end{eqnarray}

\noindent This implies that $\frac{j(j+1)}{j(j+1)-3} \Phi[J_z^2-I]
= J_z^2$ and $\Phi^{-1}[J_z^2] = \frac{j(j+1)}{j(j+1)-3}
(J_z^2-I)$. Substituting the obtained result into
formula~\eqref{Phi-selfdual-property} yields
\begin{eqnarray}
\langle 3 | \otimes \langle j,m | \, \Omega_T \, | 3 \rangle
\otimes | j,m \rangle &=& \frac{1}{j(j+1)(2j+1)} \langle j,m | \,
\Phi^{-1} [ J_z^2 ] | j,m \rangle \nonumber\\
&=& \frac{m^2-1}{(2j+1)(j^2+j-3)}.
\end{eqnarray}

\noindent If $j$ is a half-integer and $j\geq \frac{3}{2}$, then
$\langle 3 | \otimes \langle j,\frac{1}{2} | \, \Omega_T \, | 3
\rangle \otimes | j,\frac{1}{2} \rangle <0$. If $j$ is an integer
and $j\geq 2$, then $\langle 3 | \otimes \langle j,0 | \, \Omega_T
\, | 3 \rangle \otimes | j,0 \rangle <0$. Therefore, the Choi
matrix $\Omega_T$ is not positive semidefinite and $T$ is not a
channel.
\end{proof}

\section{Capacities} \label{section-capacities}

\subsection{Classical capacity}

Classical capacity~\cite{schumacher-1997,holevo-chi-1998} $C$ of a
quantum channel $\Phi$ is known to be equal to the regularized
$\chi$-capacity $C_{\chi}$, i.e., $C(\Phi) = \lim_{n \rightarrow
\infty} \frac{1}{n} C_{\chi}(\Phi^{\otimes n})$, where
$\chi$-capacity is defined by the expression
\begin{equation}
C_{\chi}(\Phi) = \sup_{\{ p_i, \rho_i \}} \left[ S \left( \sum_i
p_i \Phi[\rho_i] \right) - \sum_i p_i S \left( \Phi[\rho_i]
\right) \right]
\end{equation}

\noindent and $\{ p_i, \rho_i \}$ is an ensemble of quantum
states, in which the state $\rho_i$ is presented with the
probability $p_i$.

Further, we find $C_{\chi}(\Phi)$ for the Landau-Streater
channel~\eqref{LS-map}.

\begin{proposition} \label{prop-chi-capacity}
$\chi$-capacity $C_{\chi}(\Phi)$ of the Landau--Streater channel
$\Phi:\mathcal{H}_{2j+1} \mapsto \mathcal{H}_{2j+1}$ equals
\begin{equation} \label{C-equality}
C_{\chi}(\Phi) = \log{\frac{2j+1}{j+1}+\frac{j}{j+1}\log{j}}.
\end{equation}

\noindent If $j=1/2$, then $C(\Phi) = C_{\chi}(\Phi) = \frac{5}{3}
- \log 3$.
\end{proposition}
\begin{proof}
We exploit the fact that the Landau--Streater channel is $SU(2)$
covariant. Since the representation $U_g$ of $SU(2)$ group is
irreducible, it follows from Refs.~\cite{holevo-arxiv-2002} that
\begin{equation} \label{chi-capacity-irrep}
C_{\chi}(\Phi) = S \left( \Phi \left[\frac{1}{2j+1} I_{2j+1}
\right] \right) - \min_{\psi} S \left(\Phi[ |\psi\rangle
\langle\psi| ] \right) = \log (2j+1) - \min_{\psi} S \left(\Phi[
|\psi\rangle \langle\psi| ] \right).
\end{equation}

\noindent The minimal output entropy of $\Phi$ is given by
proposition~\ref{prop-p-norm}. Substituting~\eqref{p-norm-Phi}
into~\eqref{chi-capacity-irrep}, we get
formula~\eqref{C-equality}.

In the case $j=\frac{1}{2}$, $\chi$-capacity is known to be
additive~\cite{king-2002}, so $C(\Phi) = C_{\chi}(\Phi) =
\frac{5}{3} \log 2 - \log 3$.
\end{proof}

\subsection{Entanglement assisted capacity}

The entanglement assisted capacity $C_{\rm ea}$ quantifies the
maximal communication rate of classical information transmission
through a quantum channel $\Phi$ with the help of preshared
entanglement between the sender and receiver~\cite{bennett-1999}.
The fundamental result in quantification of the entanglement
assisted capacity is the following
formula~\cite{bennett-2002,holevo-2002}
\begin{equation} \label{c-ea-max}
C_{\rm ea}(\Phi) = \max_{\rho} \left\{ S(\rho) + S(\Phi[\rho]) -
S(\rho,\Phi) \right\},
\end{equation}

\noindent where $S(\rho,\Phi)$ is the exchange
entropy~\cite{barnum-1998}, $S(\rho,\Phi) =
S(\widetilde{\Phi}[\rho])$, with $\widetilde{\Phi}$ being the
complementary channel with respect to $\Phi$. We find the explicit
form of the entanglement assisted capacity for the
Landau--Streater channel.

\begin{proposition} \label{prop-ea-capacity}
The Landau--Streater channel $\Phi:
\mathcal{B}(\mathcal{H}_{2j+1}) \mapsto
\mathcal{B}(\mathcal{H}_{2j+1})$ has the entanglement-assisted
capacity $C_{\rm ea}(\Phi) = 2 \log (2j+1) - \log 3$.
\end{proposition}
\begin{proof}
Since $\Phi$ is irreducibly covariant by
proposition~\ref{prop-SU2-cov}, then it follows that the maximum
in~\eqref{c-ea-max} is attained on the maximally mixed input state
$\rho = \frac{1}{2j+1} I_{2j+1}$ (Ref.~\cite{holevo-2012},
Proposition 9.3). Recalling that $\Phi$ is unital and the
complementary channel $\widetilde{\Phi}$ transforms the maximally
mixed state $\frac{1}{2j+1} I_{2j+1}$ into the maximally mixed
qutrit state $\frac{1}{3} I_{3}$ by
proposition~\ref{prop-complementary-on-mixed}, we get $C_{\rm ea}
= 2 S[\frac{1}{2j+1} I_{2j+1}] - S(\frac{1}{3}I_3) = 2 \log (2j+1)
-\log 3$.
\end{proof}

\subsection{Quantum capacity}

The coherent information~\cite{schumacher-1996} for a channel
$\Phi: \mathcal{B}(\mathcal{H}) \mapsto \mathcal{B}(\mathcal{H})$
and state $\rho \in \mathcal{B}(\mathcal{H})$ is defined through
$I_{\rm c}(\rho,\Phi) = S(\Phi[\rho]) -
S(\widetilde{\Phi}[\rho])$. Maximizing coherent information over
states $\rho$ we get a ``single-letter'' quantum capacity
$Q_1(\Phi) = \max_{\rho} I_{\rm c}(\rho,\Phi)$. Quantum capacity
is known~\cite{devetak-2-2005} to be a regularized version of
$Q_1$, namely, $Q(\Phi) = \lim_{n \rightarrow \infty} \frac{1}{n}
Q_1(\Phi^{\otimes n})$. If $\Phi$ is degradable, than $Q(\Phi) =
Q_1(\Phi)$, Ref.~\cite{devetak-2005}. If $\Phi$ is antidegradable,
then $Q(\Phi)=0$, Ref.~\cite{giovannetti-fazio-2005}. If
$j=\frac{1}{2}$ or $j=1$, then the Landau--Streater channels
$\Phi: \mathcal{B}(\mathcal{H}_2) \rightarrow
\mathcal{B}(\mathcal{H}_2)$ and $\Phi: \mathcal{B}(\mathcal{H}_3)
\rightarrow \mathcal{B}(\mathcal{H}_3)$ are antidegradable, so
$Q(\Phi) = 0$. Since the Landau--Streater channel is not
antidegradable if $j \geqslant \frac{3}{2}$ by
proposition~\ref{prop-not-antidegradable}, one can expect that
$Q(\Phi) > 0$ if $j \geqslant \frac{3}{2}$. Note that $I_{\rm
c}(\rho^{\otimes n},\Phi^{\otimes n}) = n I_{\rm c}(\rho,\Phi)$,
Ref.~\cite{barnum-1998}, and therefore $Q_1(\Phi^{\otimes n}) =
\max_{\rho: \rho \in \mathcal{B}(\mathcal{H}^{\otimes n})} I_{\rm
c}(\rho,\Phi^{\otimes n}) \geq \max_{\rho: \rho \in
\mathcal{B}(\mathcal{H})} I_{\rm c}(\rho^{\otimes n},\Phi^{\otimes
n}) = n Q_1(\Phi)$. Consequently, $Q(\Phi) \geqslant Q_1(\Phi)
\geq I_{\rm c}(\rho_0,\Phi)$ for any density operator $\rho_0$.
This means that one can estimate the quantum capacity of the
Landau--Streater channel $\Phi: \mathcal{B}(\mathcal{H}_{2j+1})
\mapsto \mathcal{B}(\mathcal{H}_{2j+1})$ from below by $I_{\rm
c}(\rho_0,\Phi)$. In fact, if we fix the state $\rho_0 =
\frac{1}{2j+1} I_{2j+1}$, then $Q(\Phi) \geqslant I_{\rm
c}(\frac{1}{2j+1}I_{2j+1},\Phi) = \log(2j+1) - \log 3$. Thus, we
have just proved the following result.

\begin{proposition} \label{prop-q-capacity}
$Q(\Phi) = 0$ for the Landau--Streater channels $\Phi:
\mathcal{B}(\mathcal{H}_2) \mapsto \mathcal{B}(\mathcal{H}_2)$ and
$\Phi: \mathcal{B}(\mathcal{H}_3) \mapsto
\mathcal{B}(\mathcal{H}_3)$. If $j \geqslant \frac{3}{2}$, then
$Q(\Phi) \geqslant Q_1(\Phi) \geqslant \log(2j+1) - \log 3$ for
the Landau--Streater channel $\Phi:
\mathcal{B}(\mathcal{H}_{2j+1}) \mapsto
\mathcal{B}(\mathcal{H}_{2j+1})$.
\end{proposition}

\section{Entanglement annihilation and preservation} \label{section-entanglement}

A state $\rho \in \mathcal{B}(\mathcal{H}^{A} \otimes
\mathcal{H}^{B})$ is called separable with respect to bipartition
$A|B$ if it can be represented as the closure of convex
combination $\sum_i p_i \rho_i^{A} \otimes \rho_i^{B}$, where
$\{p_i\}$ is a probability distribution, $\rho_i^{A} \in
\mathcal{B}(\mathcal{H}^{A})$, and $\rho_i^{B} \in
\mathcal{B}(\mathcal{H}^{B})$, Ref.~\cite{werner-1989}. The
channel $\Phi: \mathcal{B}(\mathcal{H}^{A}) \mapsto
\mathcal{B}(\mathcal{H}^{A})$ is called entanglement breaking if
$\Phi \otimes {\rm Id}_k^B [\rho]$ is separable for all input
states $\rho \in \mathcal{B}(\mathcal{H}^A \otimes
\mathcal{H}^B)$, with $k = {\rm dim} \mathcal{H}^B$ being
arbitrary~\cite{horodecki-2003}. Entanglement-breaking channels
are exactly measure-and-prepare ones, and their structure is well
known~\cite{horodecki-2003}. The channel $\Lambda:
\mathcal{B}(\mathcal{H}^{A} \otimes \mathcal{H}^B) \mapsto
\mathcal{B}(\mathcal{H}^{A} \otimes \mathcal{H}^B)$ is called
entanglement annihilating if $\Lambda [\rho]$ is separable for all
input states $\rho \in \mathcal{B}(\mathcal{H}^A \otimes
\mathcal{H}^B)$~\cite{moravcikova-2010}. The structure of
entanglement annihilating channels is fully studied for local
qubit channels $\Lambda = \Phi_1 \otimes \Phi_2$~\cite{ffk-2018}
and partially studied for other classes of
channels~\cite{lami-huber-2015}. We focus on
entanglement-annihilating properties of the map $\Lambda = \Phi
\otimes \Phi$, where $\Phi$ is the Landau--Streater channel. As we
show below, $\Phi \otimes \Phi$ is not entanglement annihilating
if $j \geqslant 1$, from which it will follow that $\Phi$ is not
entanglement breaking and $\Phi \otimes \Phi$ is not absolutely
separating~\cite{fmj-2017}.

\begin{proposition} \label{prop-ent-ann}
The second tensor power of the Landau--Streater channel $\Phi:
\mathcal{B}(\mathcal{H}_{2j+1}) \mapsto
\mathcal{B}(\mathcal{H}_{2j+1})$, $\Phi \otimes \Phi$, is
entanglement annihilating if $j=\frac{1}{2}$ and is not
entanglement annihilating for all $j \geqslant 1$.
\end{proposition}
\begin{proof}
The case $j=\frac{1}{2}$ corresponds to the qubit depolarizing
channel with depolarization parameter $q=-\frac{1}{3}$.
Entanglement annihilation by $\Phi \otimes \Phi$ in this case is
proved in Ref.~\cite{moravcikova-2010}.

Let $j \geqslant 1$. In what follows, we prove that $\Phi \otimes
\Phi$ is not entanglement annihilating by presenting a bipartite
entangled state, which remains entangled after the action of $\Phi
\otimes \Phi$. Consider the vector $| \phi \rangle \in
\mathcal{H}_{2j+1} \otimes \mathcal{H}_{2j+1}$ of the form
\begin{equation} \label{schmidt-2}
|\phi\rangle = \frac{1}{\sqrt{2}} \left( |j,j\rangle |j,j\rangle +
|j,-j\rangle |j,-j\rangle \right),
\end{equation}

\noindent where $|j,m\rangle$ denotes the spin-$j$ state vector
corresponding to the definite spin projection $m$ onto $z$ axis,
$J_z |j,m\rangle = m |j,m\rangle$, $m=j,j-1,\ldots,-j$. Let $\top$
be the transposition in the basis $\{ | j,m \rangle \}$. Since
$\top \circ \Phi[X] = \Phi[X^{\top}]$ for the Landau--Streater
channel $\Phi$, then the partially transposed output state
$\Phi\otimes(\top \circ \Phi) [|\phi\rangle\langle\phi|]$ is given
by formula
\begin{eqnarray}
&& 2 \, \Phi\otimes(\top \circ \Phi) [|\phi\rangle\langle\phi|] =
\Phi[|j,j\rangle\langle j,j|]\otimes\Phi[|j,j\rangle\langle j,j|]
+ \Phi[|j,-j\rangle\langle
j,-j|]\otimes\Phi[|j,-j\rangle\langle j,-j|] \nonumber\\
&& + \Phi[|j,j\rangle\langle j,-j|]\otimes\Phi[|j,-j\rangle\langle
j,j|] + \Phi[|j,-j\rangle\langle
j,j|]\otimes\Phi[|j,j\rangle\langle j,-j|].
\end{eqnarray}

\noindent Using the channel representation $\Phi[X] =
[j(j+1)]^{-1} \left( \frac{1}{2}J_- X J_+ + \frac{1}{2} J_+ X J_-
+ J_z X J_z \right)$ and formula~\eqref{J-pm}, we get
\begin{eqnarray}
\Phi[|j,\pm j\rangle \langle j, \pm j|] &=& \frac{j}{j+1} |j,\pm
j\rangle \langle j, \pm j| + \frac{1}{j+1} |j, \pm j \mp 1\rangle
\langle j,
\pm j \mp 1|, \\
\Phi[|j, \pm j\rangle\langle j, \mp j|] &=& -\frac{j}{j+1} |j, \pm
j \rangle \langle j,\mp j|.
\end{eqnarray}

\noindent If $j \geqslant 1$, then the supports of operators
$\Phi[|j,j\rangle\langle j,j|]\otimes\Phi[|j,j\rangle\langle j,j|]
+ \Phi[|j,-j\rangle\langle j,-j|]\otimes\Phi[|j,-j\rangle\langle
j,-j|]$ and $\Phi[|j,j\rangle\langle
j,-j|]\otimes\Phi[|j,-j\rangle\langle j,j|] +
\Phi[|j,-j\rangle\langle j,j|]\otimes\Phi[|j,j\rangle\langle
j,-j|]$ are orthogonal. Moreover, the operator
\begin{eqnarray}
&& \Phi[|j,j\rangle\langle j,-j|] \otimes \Phi[|j,-j\rangle\langle
j,j|] + \Phi[|j,-j\rangle\langle
j,j|] \otimes \Phi[|j,j\rangle\langle j,-j|] \nonumber\\
&& = \frac{j^2}{(j+1)^2} \left( |j, j \rangle \langle j, -j|
\otimes |j, -j \rangle \langle j, j| + |j, -j \rangle \langle j,
j| \otimes |j, j \rangle \langle j, -j| \right)
\end{eqnarray}

\noindent is not positive semidefinite as it has a the negative
eigenvalue $-\frac{j^2}{(j+1)^2}$. Therefore, the partially
transposed state $\Phi\otimes(\top \circ \Phi)
[|\phi\rangle\langle\phi|]$ is not positive semidefinite and
$\Phi\otimes\Phi [|\phi\rangle\langle\phi|]$ is entangled by the
Peres--Horodecki criterion~\cite{peres-1996,horodecki-1996}.
\end{proof}

\begin{proposition}
The Landau--Streater channel $\Phi:
\mathcal{B}(\mathcal{H}_{2j+1}) \mapsto
\mathcal{B}(\mathcal{H}_{2j+1})$ is entanglement breaking if
$j=\frac{1}{2}$ and is not entanglement breaking if $j \geqslant
1$.
\end{proposition}
\begin{proof}
If $j=\frac{1}{2}$, then the Landau--Streater channel reduces to a
depolarizing qubit channel with depolarization parameter
$q=-\frac{1}{3}$. Such a channel is known to be entanglement
breaking~\cite{ruskai-2003}.

Let $j \geqslant 1$. Suppose that the Landau--Streater channel
$\Phi: \mathcal{B}(\mathcal{H}_{2j+1}) \mapsto
\mathcal{B}(\mathcal{H}_{2j+1})$ is entanglement breaking, then
$\Phi \otimes \Phi$ must be entanglement annihilating by
construction~\cite{moravcikova-2010}. By
proposition~\ref{prop-ent-ann} it is not the case. This
contradiction implies that $\Phi$ is not entanglement breaking.
\end{proof}

\section{Conclusions} \label{section-conclusions}

The channel~\eqref{LS-map} has been originally constructed as an
example of a unital completely positive and trace preserving map,
which is extremal in the set of channels if $j \geqslant 1$ and,
consequently, is not random unitary. By construction, the example
of Landau and Streater is $SU(2)$ covariant for all $j$ and,
surprisingly, is globally unitarily covariant if $j=\frac{1}{2}$
and $j=1$. We have proved that for $j>1$ the Landau--Streater
channels is not $U(2j+1)$ covariant, so global unitary covariance
is a peculiar property of spin-$\frac{1}{2}$ and spin-$1$ maps.
Using the theory of angular momentum, we have explicitly found the
spectrum of the Landau--Streater map in
proposition~\ref{prop-spectrum-LS} and pointed out that $\Phi$
always has negative eigenvalues. Negativity of those eigenvalues
indicates that $\Phi$ cannot be obtained as a result of Hermitian
Markovian quantum dynamics.

We have found the Stinespring dilation of the Landau--Streater
channel, which reveals its physical realization. The
Landau--Streater channel can be implemented as a result of the
controlled interaction between a spin-$j$ particle (system) and a
spin-$1$ particle (environment). The partial trace over
environment results in the Landau--Streater channel $\Phi$,
whereas the partial trace over system results in the complementary
channel $\widetilde{\Phi}$. The important property of the
complementary channel is its action on the maximally mixed input
state, which we have established in
proposition~\ref{prop-complementary-on-mixed}. If $j=\frac{1}{2}$,
then the Landau--Streater channel is antidegradable but not
degradable. If $j=1$, the Landau--Streater channel is unitary
equivalent to the Werner--Holevo channel, so in this case $\Phi$
is both degradable and antidegradable. For larger spins ($j
\geqslant \frac{3}{2}$) the Landau--Streater channel is neither
degradable nor antidegradable.

Using the theory of angular momentum, we find the minimal output
entropy of the Landau--Streater channel in
proposition~\ref{prop-p-norm}. Combining this result with $SU(2)$
covariance, we have managed to calculate the $\chi$-capacity
(proposition~\ref{prop-chi-capacity}) and the
entanglement-assisted capacity
(proposition~\ref{prop-ea-capacity}). Also, we have estimated the
lower bound on quantum capacity of the Landau--Streater channel
(proposition~\ref{prop-q-capacity}).

We have explored the entanglement dynamics induced by the
Landau--Streater channel. The channel is shown to be entanglement
breaking if and only if $j=\frac{1}{2}$. The channel's second
tensor power $\Phi \otimes \Phi$ does not annihilate entanglement
for any $j \geqslant 1$. We have constructed the state with
Schmidt rank $2$, formula~\eqref{schmidt-2}, which remains
entangled when affected by $\Phi \otimes \Phi$.

Finally, we have discussed the multiplicativity property of the
maximal $p$-norms for the Landau--Streater channel and conjectured
multiplicativity of the maximal $2$-norms with respect to the
second tensor power of the channel.

\begin{acknowledgements}
The authors thank the anonymous referee for helpful suggestions to
improve the quality of the paper, pointing out misprints, and
proposing a proof for the fact that ${\rm
rank}\Omega_{\widetilde{\Phi}} = 2j+1$ in Proposition~8. The study
in Sec. II was supported by Russian Science Foundation under
Project No. 16-11-00084. The results of Secs. III--VI were
obtained by S.N.F., supported by Russian Science Foundation under
Project No. 17-11-01388, and performed at the Steklov Mathematical
Institute of Russian Academy of Sciences.
\end{acknowledgements}


\begin{thebibliography}{99}

\bibitem{holevo-2012}
A. S. Holevo, {\it Quantum Systems, Channels, Information} (Walter
de Gruyter, Berlin, 2012).

\bibitem{Landau-Streater}
L. J. Landau and R. F. Streater, On Birkhoff's theorem for doubly
stochastic completely positive maps of matrix algebras, Lin.
Algebra Appl. {\bf 193}, 107 (1993).

\bibitem{Varshalovich}
D. A. Varshalovich, A. N. Moskalev, and V. K. Khersonskii, {\it
Theory of Angular Momentum} (World Scientific, Singapore, 1988).

\bibitem{buzek-1999}
V. Bu\v{z}ek, M. Hilery and R. F. Werner, Optimal manipulations
with qubits: Universal-NOT gate, Phys. Rev. A {\bf 60}, R2626
(1999).

\bibitem{fm-2018}
S. N. Filippov and K. Y. Magadov, Spin polarization-scaling
quantum maps and channels, Lobachevskii Journal of Mathematics
{\bf 39}, 65 (2018).

\bibitem{wh-2002}
R. F. Werner and A. S. Holevo, Counterexample to an additivity
conjecture for output purity of quantum channels, J. Math. Phys.
{\bf 43}, 4353 (2002).

\bibitem{holevo-1993}
A. S. Holevo, A note on covariant dynamical semigroups, Rep. Math.
Phys. {\bf 32}, 211 (1993).

\bibitem{holevo-1996}
A. S. Holevo, Covariant quantum Markovian evolutions, J. Math.
Phys. {\bf 37}, 1812 (1996).

\bibitem{datta-2016}
N. Datta, M. Tomamichel, and M. M. Wilde, On the second-order
asymptotics for entanglement-assisted communication, Quantum Inf.
Process. {\bf 15}, 2569 (2016).

\bibitem{mozrzymas-2017}
M. Mozrzymas, M. Studzi\'{n}ski, and N. Datta, Structure of
irreducibly covariant quantum channels for finite groups, J. Math.
Phys. {\bf 58}, 052204 (2017).

\bibitem{nuwairan-2014}
M. Al Nuwairan, The extreme points of SU(2)-irreducibly covariant
channels, Int. J. Math. {\bf 25}, 1450048 (2014).

\bibitem{datta-2004}
N. Datta, Multiplicativity of maximal $p$-norms in Werner-Holevo
channels for $1 \leq p \leq 2$, arXiv:quant-ph/0410063.

\bibitem{datta-2006}
N. Datta, A. S. Holevo, and Y. Suhov, Additivity for transpose
depolarizing channels, Int. J. Quantum Inform. {\bf 04}, 85
(2006).

\bibitem{king-depol-2003}
C. King, The capacity of the quantum depolarizing channel, IEEE
Trans. Inf. Theory {\bf 49}, 221 (2003).

\bibitem{wolf-2008}
M. M. Wolf and J. I. Cirac, Dividing Quantum Channels, Commun.
Math. Phys. {\bf 279}, 147 (2008).

\bibitem{fpmz-2017}
S. N. Filippov, J. Piilo, S. Maniscalco, and M. Ziman,
Divisibility of quantum dynamical maps and collision models, Phys.
Rev. A {\bf 96}, 032111 (2017).

\bibitem{fm-2010}
S. N. Filippov and V. I. Man'ko, Inverse spin-s portrait and
representation of qudit states by single probability vectors, J.
Russ. Laser Res. {\bf 31}, 32 (2010).

\bibitem{fm-2009}
S. N. Filippov and V. I. Man'ko, Spin tomography and star-product
kernel for qubits and qutrits, J. Russ. Laser Res. {\bf 30}, 129
(2009).

\bibitem{chruscinski-macchiavello-maniscalco-2017}
D. Chru\'{s}ci\'{n}ski, C. Macchiavello, and S. Maniscalco,
Detecting non-Markovianity of quantum evolution via spectra of
dynamical maps, Phys. Rev. Lett. {\bf 118}, 080404 (2017).

\bibitem{Tong}
J. E. Pe\v{c}ari\'{c}, F. Proschan, and Y. L. Tong, {\it Convex
functions, partial orderings, and statistical applications},
Mathematics in Science and Engineering, vol. 187 (Academic Press,
Boston, 1992).

\bibitem{michalakis-2007}
S. Michalakis, Multiplicativity of the maximal output 2-norm for
depolarized Werner-Holevo channels, J. Math. Phys. {\bf 48},
122102 (2007).

\bibitem{king-2002}
C. King, Additivity for unital qubit channels, J. Math. Phys. {\bf
43}, 4641 (2002).

\bibitem{nielsen-2000}
M. A. Nielsen and I. L. Chuang, {\it Quantum Computation and
Quantum Information} (Cambridge University Press, Cambridge,
2000).

\bibitem{holevo-2005}
A. S. Holevo, Complementary channels and the additivity problem,
Theory Probab. Appl. {\bf 51}, 92 (2007).

\bibitem{king-2007}
C. King, K. Matsumoto, M. Nathanson, and M. B. Ruskai, Properties
of conjugate channels with applications to additivity and
multiplicativity, Markov Process and Related Fields {\bf 13}, 391
(2007).

\bibitem{Cubitt}
T. S. Cubitt, M. B. Ruskai, and G. Smith, The structure of
degradable quantum channels, J. Math. Phys. {\bf 49}, 102104
(2008).

\bibitem{ruskai-2003}
M. B. Ruskai, Qubit entanglement breaking channels, Rev. Math.
Phys. {\bf 15}, 643 (2003).

\bibitem{Choi}
M.-D. Choi, Completely positive linear maps on complex matrices,
Lin. Algebra Appl. {\bf 10}, 285 (1975).

\bibitem{schumacher-1997}
B. Schumacher and M. D. Westmoreland, Sending classical
information via noisy quantum channels, Phys. Rev. A {\bf 56}, 131
(1997).

\bibitem{holevo-chi-1998}
A. S. Holevo, The capacity of quantum channel with general signal
states, IEEE Trans. Inform. Theory {\bf 44}, 269 (1998).

\bibitem{holevo-arxiv-2002}
A. S. Holevo, Remarks on the classical capacity of quantum
channel, arXiv:quant-ph/0212025.

\bibitem{bennett-1999}
C. H. Bennett, P. W. Shor, J. A. Smolin, A. V. Thapliyal,
Entanglement-assisted classical capacity of noisy quantum channel,
Phys. Rev. Lett. {\bf 83}, 3081 (1999).

\bibitem{bennett-2002}
C. H. Bennett, P. W. Shor, J. A. Smolin and A. V. Thapliyal,
Entanglement-assisted capacity of a quantum channel and the
reverse Shannon theorem, IEEE Trans. Inf. Theory {\bf 48}, 2637
(2002).

\bibitem{holevo-2002}
A. S. Holevo, On entanglement-assisted classical capacity, J.
Math. Phys. {\bf 43}, 4326 (2002).

\bibitem{barnum-1998}
H. Barnum, M. A. Nielsen, and B. Schumacher, Information
transmission through a noisy quantum channel, Phys. Rev. A {\bf
57}, 4153 (1998).

\bibitem{schumacher-1996}
B. Schumacher and M. A. Nielsen, Quantum data processing and error
correction, Phys. Rev. A {\bf 54}, 2629 (1996)

\bibitem{devetak-2-2005}
I. Devetak, The private classical capacity and quantum capacity of
a quantum channel, IEEE Trans. Inf. Theory {\bf 51}, 44 (2005).

\bibitem{devetak-2005}
I. Devetak and P. W. Shor, The capacity of a quantum channel for
simultaneous transmission of classical and quantum information,
Commun. Math. Phys. {\bf 256}, 287 (2005).

\bibitem{giovannetti-fazio-2005}
V. Giovannetti and R. Fazio, Information-capacity description of
spin-chain correlations, Phys. Rev. A {\bf 71}, 032314 (2005).

\bibitem{werner-1989}
R. F. Werner, Quantum states with Einstein-Podolsky-Rosen
correlations admitting a hidden-variable model, Phys. Rev. A {\bf
40}, 4277 (1989).

\bibitem{horodecki-2003} M. Horodecki, P. W. Shor, and M. B. Ruskai,
Entanglement breaking channels, Rev. Math. Phys. {\bf 15}, 629
(2003).

\bibitem{moravcikova-2010}
L. Morav\v{c}\'{i}kov\'{a} and M. Ziman, Entanglement-annihilating
and entanglement-breaking channels, J. Phys. A: Math. Theor. {\bf
43}, 275306 (2010).

\bibitem{ffk-2018}
S. N. Filippov, V. V. Frizen, and D. V. Kolobova, Ultimate
entanglement robustness of two-qubit states against general local
noises, Phys. Rev. A {\bf 97}, 012322 (2018).

\bibitem{lami-huber-2015}
L. Lami and M. Huber, Bipartite depolarizing maps, J. Math. Phys.
{\bf 57}, 092201 (2016).

\bibitem{fmj-2017}
S. N. Filippov, K. Y. Magadov, and M. A. Jivulescu, Absolutely
separating quantum maps and channels, New J. Phys. {\bf 19},
083010 (2017).

\bibitem{peres-1996}
A. Peres, Separability criterion for density matrices, Phys. Rev.
Lett. {\bf 77}, 1413 (1996).

\bibitem{horodecki-1996}
M. Horodecki, P. Horodecki, and R. Horodecki, Separability of
mixed states: necessary and sufficient conditions, Phys. Lett. A
{\bf 223}, 1 (1996).

\end{thebibliography}
\end{document}